\documentclass[
  aps,
  pra,
  reprint,
  amsmath,
  amssymb,
  superscriptaddress
]{revtex4-2}

\usepackage{graphicx}
\usepackage{dcolumn}
\usepackage{bm}
\usepackage{xcolor}
\usepackage{booktabs}
\usepackage{mathtools}
\usepackage{physics}
\usepackage{tikz}
\usepackage{amsmath, amsfonts, amsthm}
\usetikzlibrary{positioning,calc,arrows.meta,shapes.geometric}

\definecolor{tnink}{HTML}{0B2D4D}
\definecolor{tnlightblue}{HTML}{A6D8F5}
\definecolor{tnblue}{HTML}{0072B2}
\definecolor{tngrey}{HTML}{808080}
\definecolor{tnlightgrey}{HTML}{D9D9D9}
\definecolor{tngreen}{HTML}{009E73}
\definecolor{tnorange}{HTML}{E69F00}
\definecolor{tnred}{HTML}{D55E00}
\definecolor{tnpurple}{HTML}{CC79A7}
\definecolor{tncyan}{HTML}{56B4E9}
\tikzset{
  leg/.style={draw=tnink, line width=1.2pt, line cap=round, line join=round},
  tensor/.style={
  rectangle,
  draw=tnink,
  line width=1.1pt,
  minimum size=8mm,
  inner sep=0pt,
  text width=8mm,
  align=center,
  fill=#1,
  text=tnink
  },
  tensor/.default=tnblue,
  op/.style={rectangle, rounded corners=1.2mm, draw=tnink, line width=1.1pt, minimum width=9mm, minimum height=7mm, inner sep=1pt, fill=#1, text=white},
  op/.default=tncyan,
  iso/.style={regular polygon, regular polygon sides=3, draw=tnink, line width=1.1pt, minimum size=10mm, inner sep=0pt, fill=#1, text=white},
  iso/.default=tnorange,
  lab/.style={font=\small, text=black},
}

\newtheorem{theorem}{Theorem}

\newtheorem{proposition}[theorem]{Proposition}

\newtheorem{definition}[theorem]{Definition}

\newcommand{\C}{\mathbb{C}}

\renewcommand{\P}{\mathbb{P}}
\newcommand{\id}{\text{id}}

            \usepackage{mathtools}

\newcommand{\panel}[5]{%
\begin{tikzpicture}
    \node[inner sep=0, outer sep=0] (img)
        {\includegraphics[width=#1]{#2}};
    \node[
        overlay,
        anchor=north east,
        font=\small\bfseries,
        inner sep=1pt,
        outer sep=0
    ] at ([xshift=#4,yshift=#5]img.north west) {#3};
\end{tikzpicture}%
}

\begin{document}


\title{Excitation spectra and rank tomography of finite MPS tangent spaces}

\author{Otto T. P. Schmidt}
\email{otto.schmidt@mis.mpg.de}

\affiliation{
INO-CNR Pitaevskii BEC Center and Dipartimento di Fisica,
Universit\`a di Trento, Trento, Italy
}

\affiliation{
Max Planck Institute for Mathematics in the Sciences,
Leipzig, Germany
}

\author{Iacopo Carusotto}

\affiliation{
INO-CNR Pitaevskii BEC Center and Dipartimento di Fisica,
Universit\`a di Trento, Trento, Italy
}


\begin{abstract}
\noindent
We formulate the matrix product state (MPS) tangent-space excitation construction for finite, non-uniform systems with open boundary conditions, using the smooth full-rank stratum of the MPS variety as the variational manifold. The resulting linear tangent ansatz yields a projected-Hamiltonian eigenproblem for approximating excitation spectra. We further introduce a rank tomography that characterizes the particle-sector expressivity of the MPS tangent space. For number-conserving systems, we resolve the Schmidt ranks of reference states by particle number and derive an explicit relation between the resulting rank profile and the dimensions of the tangent-space sectors. This allows sector completeness and parametric deficiency to be determined directly from the reference state, without explicitly constructing a tangent basis, and provides a direct diagnostic of the sectorwise expressivity of the MPS excitation ansatz. We benchmark the finite-system construction on Bose--Hubbard chains against exact diagonalization, finding accurate low-lying excitation branches while identifying sector-dependent limitations of the linear tangent ansatz.

\end{abstract}

\maketitle


\section{Introduction}\label{sec:intro}

Matrix product states (MPS) provide an efficient variational description
of ground states of one-dimensional quantum many-body systems
\cite{Schollwoeck_2010, landsberg_2012, Oseledets_2011}.
Their usefulness also extends to excitations: linear perturbations
around a variational ground state define a tangent-space ansatz for
approximating low-energy excitation spectra
\cite{haegeman_2011, Haegeman_2014, Nakatani_2014, van_damme_2021, Osborne_2025, Chen_2026}.
Here we formulate this excitation problem for finite, non-uniform MPS
with open boundary conditions, emphasizing the algebro-geometric structure
of the variational space. MPS with bond dimensions bounded by prescribed
values form a projective variety defined by rank constraints on tensor
flattenings \cite{Borovik_2025, hosten_2026}. We work on its smooth full-rank stratum and construct
the physical tangent space by removing gauge redundancy and the
reference-state direction. Projecting the Hamiltonian onto this
linear space yields an effective eigenvalue problem for excitation
energies.

At fixed bond dimensions, how many independent excitation directions
does the tangent space contain in each particle-number sector?
We address this question through a sector-resolved
\emph{tangent-space tomography} \cite{Zauner_2018}. Given a reference state, we compute ranks of the tangent space projected onto particle-number
sectors. These ranks define the
\emph{parametric deficiency}, which counts the independent directions
missing from each target sector. To explain their dependence on the
reference state, we use a particle-resolved refinement of the usual
MPS bond dimensions. The resulting
\emph{particle-resolved Schmidt-rank profile} (PRSR) records
how the Schmidt rank across each cut is distributed among
particle-number sectors of the left subsystem
\cite{Bachmayr_2022, Zauner_2018}. We derive an explicit counting formula that
determines the sector tangent dimensions from the PRSR profile
at a full-rank reference state. Thus parametric deficiency can be diagnosed directly from the reference state's particle-resolved ranks, without constructing a tangent-space basis. This relation also explains how reference states with identical bond dimensions can support different sector-resolved tangent dimensions, showing how particle-number symmetry refines the coarse bond dimensions into a structure that directly controls sectorwise tangent-space expressivity.

We illustrate these ideas with finite Bose--Hubbard chains, using
exact diagonalization to benchmark the reconstructed spectra.
The calculations accurately approximate the low-lying excitation
branches considered, including the softening of the particle-addition
mode near a ground-state sector crossing. The examples illustrate the trade-off between
spectral accuracy and computational cost as the bond dimensions
increase. They also show why bond dimensions alone do not
describe the sector-dependent limitations of the tangent ansatz.

This work combines an explicit algebro-geometric formulation of the
finite MPS tangent-space method with a diagnostic of its
particle-sector expressivity. The projected Hamiltonian provides
excitation energy estimates, while the rank tomography identifies
which sectors are completely represented. Together, these constructions relate
spectral reconstruction to the rank structure of the reference
state, while distinguishing sector completeness from the accuracy
of individual excitation energies.

The paper is organized as follows.
Section~\ref{sec:state_description} introduces the MPS variety,
its full-rank stratum, and the physical tangent space.
Section~\ref{sec:lin_tangent_method} derives the projected
tangent-space eigenproblem from the time-dependent variational principle restricted to a
linear tangent ansatz.
Section~\ref{sec:results} presents the Bose--Hubbard benchmarks,
and Section~\ref{sec:tangent_space_tomography} develops the
sector-resolved tomography, derives the counting formula connecting
PRSR to the tangent-sector dimensions, and gives illustrative examples.
Section~\ref{sec:conclusion} summarizes the results and discusses
extensions. The Appendix provides explicit examples, computational
cost estimates, and proofs of the geometric statements used in
the construction.

\section{\label{sec:state_description}Variational manifolds of MPS}

For an $N$-site system, let $\mathcal H_i=\mathbb C^{d_i}$ be
the local Hilbert space at site $i$, and collect the local dimensions
in $\boldsymbol d=(d_1,\dots,d_N)\in\mathbb Z_{>0}^N$.
The full Hilbert space is
\[
\mathcal H=\bigotimes_{i=1}^N\mathcal H_i.
\]
A pure quantum state is represented by a complex line in
$\mathcal H$, or equivalently by a point in the projective
Hilbert space $X:=\mathbb P(\mathcal H)$.
For a nonzero vector $\psi\in\mathcal H$, we write
\[
[\psi]
=
\{\lambda\psi:\lambda\in\mathbb C^\ast\}
\]
for its projective class.

Choose an ordered orthonormal basis
$\{e_{j_i}\}_{j_i=1}^{d_i}$ of each local Hilbert space.
Then every vector $\psi\in\mathcal H$ has the expansion
\[
\psi
=
\sum_{j_1=1}^{d_1}\cdots\sum_{j_N=1}^{d_N}
c_{j_1,\dots,j_N}\,
e_{j_1}\otimes\cdots\otimes e_{j_N},
\]
with coefficients $c_{j_1,\dots,j_N}\in\mathbb C$.
Thus, after fixing these bases, the state vector is identified
with its coefficient tensor.

A variational family of pure states is a subset
$\mathcal M\subseteq X$.
For MPS, we describe this family through an algebraic variety
defined by polynomial relations. We then use its smooth full-rank
stratum as the variational manifold for the tangent-space
construction. The algebraic description and the
differential-geometric treatment of variational dynamics \cite{Haegeman_2014} are
therefore complementary.

\subsection{Tensor network states and MPS}\label{section:MPS}

A general state vector requires
\[
\dim\mathcal H=\prod_{i=1}^N d_i
\]
complex coefficients. For uniform local dimension $d_i=d>1$,
this number grows as $d^N$.
Tensor network states replace the full coefficient tensor by
a contraction of local tensors. Restricting the size of these tensors gives a structured variational family.

The pattern of contractions is abstractly described by a graph \cite{landsberg_2012}.
For matrix product states, this graph is a path on $N$ vertices,
one for each site. We denote the \emph{bond dimensions} by
\[
\boldsymbol D=(D_0,D_1,\dots,D_N)\in\mathbb Z_{>0}^{N+1}
\]
and impose $D_0=D_N=1$ for open boundary conditions (OBC).
Each physical coefficient is then obtained by contracting a set of matrices.

\begin{definition}
The affine MPS parameter space is
\[
\mathcal A_{\boldsymbol D,\boldsymbol d}
:=
\prod_{i=1}^N
\bigl(
\mathbb C^{D_{i-1}}
\otimes\mathbb C^{d_i}
\otimes\mathbb C^{D_i}
\bigr).
\]
For each core $A^i$, fixing its physical index $j_i$ gives
a matrix $A^i_{j_i}\in\mathbb C^{D_{i-1}\times D_i}$.
The affine MPS contraction map
\[
\widetilde\Phi_{\boldsymbol D,\boldsymbol d}:
\mathcal A_{\boldsymbol D,\boldsymbol d}\longrightarrow\mathcal H
\]
is defined by
\begin{align*}
&\widetilde\Phi_{\boldsymbol D,\boldsymbol d}
(A^1,\dots,A^N)\\
&\quad=
\sum_{j_1=1}^{d_1}\cdots\sum_{j_N=1}^{d_N}
\bigl(A^1_{j_1}\cdots A^N_{j_N}\bigr)\,
e_{j_1}\otimes\cdots\otimes e_{j_N}.
\end{align*}
The boundary conditions $D_0=D_N=1$ ensure that each matrix
product is a scalar.
\end{definition}

The contraction map is polynomial in the core entries.
Whenever the contraction is nonzero, it also defines a
projective state. Hence the projective parametrization is
\[
\mathbb P\widetilde\Phi_{\boldsymbol D,\boldsymbol d}:
\mathcal A_{\boldsymbol D,\boldsymbol d}
\setminus
\widetilde\Phi_{\boldsymbol D,\boldsymbol d}^{-1}(0)
\longrightarrow X,
\]
with
\[
\mathbb P\widetilde\Phi_{\boldsymbol D,\boldsymbol d}
(A^1,\dots,A^N)
=
\bigl[
\widetilde\Phi_{\boldsymbol D,\boldsymbol d}
(A^1,\dots,A^N)
\bigr].
\]
This is a regular map on the complement of the
zero-contraction locus.

\paragraph{Left-canonical parametrization.}

For the tangent-space construction, we organize the tensor
parameters in \emph{left-canonical form} \cite{Schollwoeck_2010}.
We write the cores in this representation as
$M^1,\dots,M^{N-1},C$, where $C\in\mathbb C^{D_{N-1}\times d_N}$ is the final core, and $C_{j_N}$ denotes its $j_N$-th column.

We assume throughout that the bond dimensions are
\emph{admissible}:
\[
D_i\leq D_{i-1}d_i,
\qquad
D_{i-1}\leq d_iD_i,
\qquad i=1,\dots,N.
\]
These inequalities ensure that the required local ranks
are feasible.

The left-canonical condition requires the vertically stacked
matrix at each site $i<N$ to have orthonormal columns.
To express this condition, define the complex Stiefel manifold
\[
\mathrm{St}(k,n)
:=
\{A\in\mathbb C^{n\times k}:A^\ast A=\id_k\}.
\]
For $i=1,\dots,N-1$, we have
\begin{equation*}
M^i
:=
\begin{bmatrix}
M^i_1\\
\vdots\\
M^i_{d_i}
\end{bmatrix}
\in\mathrm{St}(D_i,D_{i-1}d_i),
\qquad
M^i_{j_i}\in\mathbb C^{D_{i-1}\times D_i}.
\end{equation*}
Equivalently, $(M^i)^\ast M^i
=
\sum_{j_i=1}^{d_i}
(M^i_{j_i})^\ast M^i_{j_i}
=
\id_{D_i}.$

Left-canonical normalization alone does not ensure that the
contracted state attains the prescribed Schmidt rank at every
cut. We therefore impose a second local rank condition.
For the horizontal concatenation
\[
\widehat M^i
:=
\begin{bmatrix}
M^i_1&\cdots&M^i_{d_i}
\end{bmatrix}
\in\mathbb C^{D_{i-1}\times(d_iD_i)},
\]
we require full row rank, that is, $\rank\widehat M^i=D_{i-1}$.
We denote the corresponding restricted Stiefel factor by
\begin{align*}
    &\mathrm{St}^{=}(D_i,D_{i-1}d_i) \\
&\qquad\qquad:=
\left\{
M^i\in\mathrm{St}(D_i,D_{i-1}d_i):
\rank\widehat M^i=D_{i-1}
\right\}.
\end{align*}
Together with the full-row-rank condition
$\rank C=D_{N-1}$, these conditions ensure that the contracted
state has Schmidt rank $D_\ell$ across each cut after site $\ell$.

\begin{definition}
For admissible bond dimensions, the full-rank left-canonical MPS parameter space is
\[
\mathcal P_{\boldsymbol D,\boldsymbol d}
:=
\prod_{i=1}^{N-1}
\mathrm{St}^{=}(D_i,D_{i-1}d_i)
\times
\mathbb P^{=}
\bigl(\mathbb C^{D_{N-1}\times d_N}\bigr),
\]
where
$\mathbb P^{=}(\mathbb C^{D_{N-1}\times d_N})$
denotes the projective classes of full-row-rank matrices. The left-canonical MPS map
\[
\Phi_{\boldsymbol D,\boldsymbol d}:
\mathcal P_{\boldsymbol D,\boldsymbol d}\longrightarrow X
\]
is given by
\begin{align*}
&\Phi_{\boldsymbol D,\boldsymbol d}
(M^1,\dots,M^{N-1},[C])\\
&\quad=
\left[
\sum_{j_1=1}^{d_1}\cdots\sum_{j_N=1}^{d_N}
\bigl(
M^1_{j_1}\cdots M^{N-1}_{j_{N-1}}C_{j_N}
\bigr)\,
e_{j_1}\otimes\cdots\otimes e_{j_N}
\right].
\end{align*}
\end{definition}

The additional full-rank restrictions are open conditions.
Thus each restricted Stiefel factor has the same tangent space
as the original Stiefel manifold at its retained points.
Figure~\ref{fig:TT_decomp} illustrates the parametrization for
a five-site system.

\begin{figure}[ht]
\centering
\begin{tikzpicture}[x=1cm,y=1cm]

  \begin{scope}
    \foreach \i/\c/\name in {
      0/tnlightblue/M^1,
      1/tnlightblue/M^2,
      2/tnlightblue/M^3,
      3/tnlightblue/M^4,
      4/tnpurple/C
    }{
      \pgfmathtruncatemacro{\site}{\i+1}
      \node[tensor=\c] (T\i) at (1.5*\i,0) {$\name$};
      \draw[leg] (T\i) -- ++(0,0.75)
        node[lab,above] {$d_{\site}$};
    }

    \foreach \i in {0,1,2,3}{
      \pgfmathtruncatemacro{\j}{\i+1}
      \draw[leg] (T\i) -- node[lab,above] {$D_{\j}$} (T\j);
    }
  \end{scope}

\end{tikzpicture}

\caption{A five-site MPS in left-canonical form, with physical dimensions $d_i$, bond dimensions $D_i$, and final tensor $C$.}
\label{fig:TT_decomp}
\end{figure}

\paragraph{Gauge redundancy of the parametrization.}

The left-canonical conditions do not make the representation
unique. A unitary matrix inserted on one side of an internal
bond can be canceled by its inverse on the other side,
without changing the contracted state.
These transformations form the residual gauge group of the
left-canonical parametrization.

\begin{definition}
The left-canonical MPS gauge group is
\[
G_{MPS}
:=
\prod_{i=1}^{N-1}\mathrm U(D_i).
\]
Its action
\[
g_{MPS}:
\mathcal P_{\boldsymbol D,\boldsymbol d}\times G_{MPS}
\longrightarrow
\mathcal P_{\boldsymbol D,\boldsymbol d}
\]
is specified, for $g=(G_1,\dots,G_{N-1})$, by
\begin{align*}
M^i_{j_i}
&\longmapsto
G_{i-1}^{-1}M^i_{j_i}G_i,\\
[C]
&\longmapsto
[G_{N-1}^{-1}C],
\end{align*}
where $G_0=1$ and the first transformation applies to every
physical slice $j_i$.
\end{definition}

These transformations preserve the left-canonical and full-rank
conditions. The factors $G_i$ and $G_i^{-1}$ cancel between
adjacent cores, so
\[
\Phi_{\boldsymbol D,\boldsymbol d}\bigl(g_{MPS}(p,g)\bigr)
=
\Phi_{\boldsymbol D,\boldsymbol d}(p).
\]
This residual nonuniqueness gives rise to the gauge directions
that need to be removed in the tangent-space construction.

\subsection{MPS varieties}\label{section:MPS_varieties}

We describe the MPS variational family through polynomial relations
in the state coefficients. Recent work in algebraic geometry has
established various properties of these MPS varieties, also known
as tensor train (TT) varieties in the open-boundary setting.
Dimension results for tensor network varieties, including MPS,
are given in \cite{Bernardi_2022}. Expressions for the degrees
of MPS varieties were obtained in
\cite{rosana_2026}, and generators of their prime ideals
were established in \cite{hosten_2026}.
The latter result allows a
parametrization-independent description of the tangent space
at a full-rank state.

\begin{definition}
The projective MPS variety is \cite{Borovik_2025},
\[
V_{\boldsymbol D,\boldsymbol d}
:=
\overline{
\operatorname{Im}(\P\widetilde{\Phi}_{\boldsymbol D,\boldsymbol d})
}^{\,\mathrm{Zar}}
\subseteq X.
\]
\end{definition}

For MPS with open boundary conditions, this closure admits an
explicit description in terms of flattening ranks.
For a nonzero tensor $T\in\mathcal H$ and each $k=1,\dots,N-1$,
consider the flattening across the cut after site $k$,
\[
T^{(k)}:
\bigotimes_{i=k+1}^N(\C^{d_i})^\ast
\longrightarrow
\bigotimes_{i=1}^k\C^{d_i}.
\]
If $T$ has an MPS representation with bond dimensions
$\boldsymbol D$, then the flattening factors through the
$D_k$-dimensional bond space. Hence $\rank(T^{(k)})\leq D_k$ for all $k=1,\dots,N-1$.
Conversely, these bounds are sufficient: successive rank
factorizations produce an MPS decomposition with bond dimensions
bounded by $\boldsymbol D$.
Each rank bound is equivalent to the vanishing of all
$(D_k+1)\times(D_k+1)$ minors of the corresponding flattening.
These are homogeneous polynomial equations, so the image is
already Zariski closed in projective space.

\begin{proposition}\label{prop:rank_description_variety}
The projective MPS variety admits the description
\[
V_{\boldsymbol D,\boldsymbol d}
=
\left\{
[T]\in X
\,\middle|\,
\rank(T^{(k)})\leq D_k,\,k=1,\dots,N-1
\right\}.
\]
\end{proposition}

A proof is given in \cite{Borovik_2025}.
The variety $V_{\boldsymbol D,\boldsymbol d}$ decomposes into
strata according to its flattening ranks.
We define the \emph{full-rank MPS stratum} as the subset on which
each flattening attains its prescribed rank:
\[
V^{=}_{\boldsymbol D,\boldsymbol d}
:=
\left\{
[T]\in X
\,\middle|\,
\rank(T^{(k)})=D_k,\,k=1,\dots,N-1
\right\}.
\]

Tensors of fixed TT rank form a smooth embedded manifold
\cite{Holtz_2012, Haegeman_2014}.
Moreover, the full-rank stratum is Zariski open in
$V_{\boldsymbol D,\boldsymbol d}$: its complement is the union
of the closed subsets on which at least one flattening rank drops.
We record the corresponding projective statement.

\begin{proposition}\label{prop:smoothness_full_rank}
The full-rank stratum $V^{=}_{\boldsymbol D,\boldsymbol d}$
is a smooth Zariski-open subset of the projective variety
$V_{\boldsymbol D,\boldsymbol d}$.
\end{proposition}
\begin{proof}
See Appendix~\ref{app:proofs_smoothness_full_rank}.
\end{proof}

Note that the unrestricted contraction map provides the polynomial parametrization used to define the MPS variety, while the left-canonical parametrization is used for differential-geometric constructions on its smooth full-rank stratum. Both describe the same full-rank physical states.
We use $\mathcal M=V^{=}_{\boldsymbol D,\boldsymbol d}$ as the
variational manifold in the following.
The full-rank conditions imposed on
$\mathcal P_{\boldsymbol D,\boldsymbol d}$ in
Section~\ref{section:MPS} ensure that the image of $\Phi_{\boldsymbol D,\boldsymbol d}$ lies in this
stratum. Conversely, successive singular value decompositions
give every element of $V^{=}_{\boldsymbol D,\boldsymbol d}$
a left-canonical representation satisfying these conditions.
Thus
\[
\Phi_{\boldsymbol D,\boldsymbol d}:
\mathcal P_{\boldsymbol D,\boldsymbol d}
\longrightarrow
V^{=}_{\boldsymbol D,\boldsymbol d}
\]
is surjective. More precisely, $\Phi_{\boldsymbol D,\boldsymbol d}$ is a
smooth surjective submersion \cite{Haegeman_2014, rosana_2026}.
Consequently, every tangent vector to the full-rank MPS stratum
has a lift to the parameter tangent space.
This property underlies the horizontal tangent-space construction
in the next subsection.

\subsection{Tangent space of the MPS variety}
\label{section:tangent_space_MPS}

We now separate parameter variations that change the physical state
from those that only change its MPS representation.
The latter are tangent to the gauge orbits and are removed by
choosing an orthogonal complement in the parameter tangent space.

Fix
\[
p=(M^1,\dots,M^{N-1},[C])
\in\mathcal P_{\boldsymbol D,\boldsymbol d},
\]
and choose a representative $C$ with $\|C\|_F=1$.
We represent tangent vectors to the final projective factor
by matrices $\delta C$ satisfying
\[
\langle C,\delta C\rangle_F
:=
\operatorname{tr}(C^\ast\delta C)=0.
\]
With this convention, we write a parameter tangent vector as
$\delta p=(\delta M^1,\dots,\delta M^{N-1},\delta C)$.

We regard $\mathcal P_{\boldsymbol D,\boldsymbol d}$ as a real
Riemannian manifold. Since the Stiefel constraints
involve complex conjugation, their tangent spaces are not
closed under multiplication by $i$.
We equip the parameter space with the product of the real
Frobenius metric on the Stiefel factors and the Fubini--Study
metric on the projective factor.
Using the chosen representatives, this metric is
\[
g_p(\delta p,\delta \tilde p)
=
\sum_{i=1}^{N-1}
\operatorname{Re}\operatorname{tr}
\bigl((\delta M^i)^\ast\delta \tilde M^i\bigr)
+
\operatorname{Re}\operatorname{tr}
\bigl(\delta C^\ast\delta \tilde C\bigr).
\]
All orthogonal complements below are taken with respect to
this real metric.

The \emph{vertical tangent space} is
$\ker(D_p\Phi_{\boldsymbol D,\boldsymbol d})$.
Its orthogonal complement is the \emph{horizontal tangent space},
giving the decomposition
\[
T_p\mathcal P_{\boldsymbol D,\boldsymbol d}
=
\ker(D_p\Phi_{\boldsymbol D,\boldsymbol d})^\perp
\oplus
\ker(D_p\Phi_{\boldsymbol D,\boldsymbol d}).
\]
Since $\Phi_{\boldsymbol D,\boldsymbol d}$ is a smooth submersion,
\[
\ker(D_p\Phi_{\boldsymbol D,\boldsymbol d})
=
T_p\!\left(
\Phi_{\boldsymbol D,\boldsymbol d}^{-1}
\bigl(\Phi_{\boldsymbol D,\boldsymbol d}(p)\bigr)
\right).
\]
On the full-rank parameter space, the fibers are gauge orbits.
Thus the vertical space consists precisely of the infinitesimal
gauge transformations \cite{Haegeman_2014}.

Differentiating the gauge action at the identity gives a vertical
tangent vector $\dot p=(\dot M^1,\dots,\dot M^{N-1},\dot C)$.
In the chosen projective tangent representatives, its components are
\begin{align}\label{eq:vertical_directions_explicit}
\dot M^1_{j_1}
&=M^1_{j_1}X_1, \nonumber\\
\dot M^i_{j_i}
&=M^i_{j_i}X_i-X_{i-1}M^i_{j_i},
\qquad i=2,\dots,N-1, \nonumber\\
\dot C
&=-X_{N-1}C
  +C\,\operatorname{tr}(C^\ast X_{N-1}C)
,
\end{align}
where $X_i:=\dot G_i(0)\in\mathfrak u(D_i)$ and $X_i^\ast=-X_i$.

The gauge action is free, so its infinitesimal generators
parametrize the vertical tangent space without redundancy.
Consequently,
\[
\ker(D_p\Phi_{\boldsymbol D,\boldsymbol d})
\cong
\bigoplus_{i=1}^{N-1}\mathfrak u(D_i)
\]

A parameter tangent vector $\delta p$ is
horizontal precisely when
\[
g_p(\delta p,\dot p)=0
\qquad
\text{for every }
\dot p\in\ker(D_p\Phi_{\boldsymbol D,\boldsymbol d}).
\]
Substituting the infinitesimal gauge transformations gives
the following explicit conditions.

\begin{proposition}\label{prop:horizontal_conditions}
Let $\delta p\in
T_p\mathcal P_{\boldsymbol D,\boldsymbol d}$,
with $\delta C$ represented so that
$\operatorname{tr}(C^\ast\delta C)=0$.
This tangent vector is horizontal if and only if the matrices
\[
Y_i
:=
\sum_{j_i=1}^{d_i}
(\delta M^i_{j_i})^\ast M^i_{j_i}
-
\sum_{j_{i+1}=1}^{d_{i+1}}
M^{i+1}_{j_{i+1}}
(\delta M^{i+1}_{j_{i+1}})^\ast,
\]
for $i=1,\dots,N-2$, and
\[
Y_{N-1}
:=
\sum_{j_{N-1}=1}^{d_{N-1}}
(\delta M^{N-1}_{j_{N-1}})^\ast M^{N-1}_{j_{N-1}}
-
\sum_{j_N=1}^{d_N}
C_{j_N}(\delta C_{j_N})^\ast
\]
are Hermitian.
\end{proposition}
\begin{proof}
See Appendix~\ref{app:proofs_horizontal_conditions}.
\end{proof}

\section{Excitations from the MPS tangent space}
\label{sec:lin_tangent_method}

The preceding sections define the full-rank MPS manifold and its
physical tangent space. We now use this tangent space to approximate
quantum dynamics and construct a linear ansatz for excitation spectra.
We first introduce the time-dependent variational principle (TDVP)
\cite{kramer1981} and then derive the projected tangent-space
eigenproblem. These constructions build on established MPS
tangent-space methods
\cite{kloss_2011, haegeman_2011, Nakatani_2014}.

\subsection{Time-dependent variational principle}\label{sec:TDVP}

The TDVP approximates Schr\"odinger dynamics while constraining the state to a prescribed variational manifold.
For a Hamiltonian $H$, the exact Schr\"odinger equation
is
\[
\partial_t\psi(t)=-iH\psi(t),
\]
where we set $\hbar=1$.
In the variational approximation, we require
$[\psi(t)]\in V^{=}_{\boldsymbol D,\boldsymbol d}$.
At each state, the TDVP selects the allowed tangent velocity
closest to the exact Schr\"odinger velocity.

Fix $p\in\mathcal P_{\boldsymbol D,\boldsymbol d}$ with
$[\psi]=\Phi_{\boldsymbol D,\boldsymbol d}(p)$, and choose a normalized
Hilbert-space representative $\psi$ of this projective state.
We write
\[
\psi^\perp
:=
\{v\in\mathcal H:\langle\psi,v\rangle=0\}
\]
for the orthogonal complement of the line
$\mathbb C\psi$.

A parameter variation is physically relevant only after removing two
redundancies. First, we restrict to the horizontal parameter space
\[
H_p
:=
\ker(D_p\Phi_{\boldsymbol D,\boldsymbol d})^\perp,
\]
which removes infinitesimal gauge transformations of the MPS
representation. Second, since variations proportional to $\psi$ change
only the representative of the same projective state, we project the
induced Hilbert-space variation onto $\psi^\perp$.

For $\delta p\in H_p$, differentiating the tensor contraction gives
the Hilbert-space vector $w(\delta p)$ with coefficients
\[
w(\delta p)_{j_1,\dots,j_N}
=
\sum_{i=1}^{N-1}
L_{i-1}\,\delta M^i_{j_i}\,R_{i+1}
+
L_{N-1}\,\delta C_{j_N},
\]
where
\[
L_i=M^1_{j_1}\cdots M^i_{j_i},
\qquad
R_i=M^i_{j_i}\cdots M^{N-1}_{j_{N-1}}C_{j_N},
\]
with $L_0=1$ and $R_N=C_{j_N}$.
We define the corresponding physical tangent vector by
\[
\delta\psi
=
Q_\psi w(\delta p),
\qquad
Q_\psi:=I-|\psi\rangle\langle\psi|.
\]

The \emph{physical tangent space} at $\psi$ is therefore
\[
T_\psi
:=
\left\{
Q_\psi w(\delta p):
\delta p\in H_p
\right\}
\subseteq\psi^\perp.
\]
By the submersion property of
$\Phi_{\boldsymbol D,\boldsymbol d}$, this space is naturally
isomorphic to the projective tangent space,
\[
T_\psi
\cong
T_{[\psi]}V^{=}_{\boldsymbol D,\boldsymbol d}.
\]
Thus the horizontal condition removes parameter-space gauge
redundancy, while $Q_\psi$ removes the unphysical reference-state
direction in Hilbert space.

\paragraph{Orthogonal projection.}

The exact velocity $-iH\psi$ generically does not lie on $T_\psi$.
We therefore choose
\[
v^\ast
=
\underset{v\in T_\psi}{\arg\min}
\|v+iH\psi\|^2.
\]
This is the orthogonal projection of $-iH\psi$ onto $T_\psi$.
The minimizing residual is orthogonal to every allowed tangent
direction:
\[
\langle\tau,v^\ast+iH\psi\rangle=0
\qquad
\text{for all }\tau\in T_\psi.
\]

We define the variational trajectory by using the minimizing
vector as its velocity $\partial_t\psi:=v^\ast$.
Because $v^\ast\in\psi^\perp$, this chooses a normalized
representative with $\langle\psi,\partial_t\psi\rangle=0$.
We can equivalently write this as the 
\emph{Dirac--Frenkel condition} \cite{yuan2019}
\begin{equation}\label{eq:dirac_frenkel}
\langle\tau,(i\partial_t-H)\psi\rangle=0
\qquad
\text{for all }\tau\in T_\psi.
\end{equation}

The same condition can be obtained from the TDVP Lagrangian \cite{Krutitsky_2011}
\begin{equation}\label{eq:TDVP_lagrangian}
\mathcal L(\psi,\dot\psi)
=
\frac{i}{2}
\bigl(
\langle\psi,\dot\psi\rangle
-
\langle\dot\psi,\psi\rangle
\bigr)
-
\langle\psi,H\psi\rangle.
\end{equation}
Its Euler--Lagrange equations in the physical
tangent directions give \eqref{eq:dirac_frenkel}.

Thus the TDVP replaces the Schr\"odinger velocity at each state
by its closest physical tangent velocity. 
Figure~\ref{fig:TDVP} illustrates the tangent-space projection
schematically using a real rank-one tensor cone.

\begin{figure}
    \centering
    \includegraphics[width=\linewidth]
    {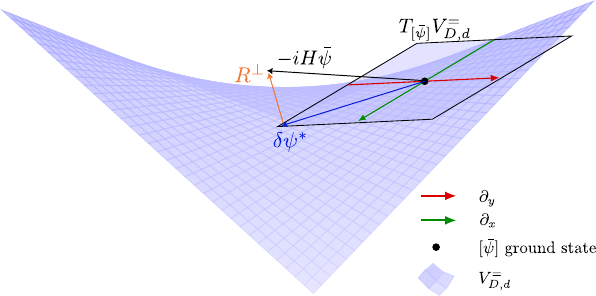}
    \caption{Schematic illustration of tangent-space projection
    in TDVP using $\delta\psi^* = v^*$ and $R^\perp = v^\ast+iH\psi$. A real rank-one tensor cone associated with
    $V^{=}_{(1,1,1),(2,2)}$ is used for illustration.}
    \label{fig:TDVP}
\end{figure}

\subsection{Extraction of excitation spectra}
\label{sec:linear_tangent_mps}

We now use the physical tangent space at a variational ground-state (or reference state)
approximation as a fixed linear space of trial excitation vectors.
Applying the TDVP action to fluctuations in this space yields a
Hermitian eigenvalue problem for excitation energies.

Let $\bar\psi$ be a normalized full-rank reference state with
\[
[\bar\psi]=\Phi_{\boldsymbol D,\boldsymbol d}(\bar p),
\qquad
\bar p\in\mathcal P_{\boldsymbol D,\boldsymbol d},
\]
obtained by minimizing the Rayleigh--Ritz quotient
\[
f(\psi)
=
\frac{\langle\psi,H\psi\rangle}{\langle\psi,\psi\rangle},
\qquad
[\psi]\in V^{=}_{\boldsymbol D,\boldsymbol d}.
\]
Such a ground-state approximation can be computed
using MPS methods such as DMRG
\cite{Schollwoeck_2010, White1992}.
We denote its variational energy and the shifted Hamiltonian by $\omega_0:=\langle\bar\psi,H\bar\psi\rangle$ and $H':=H-\omega_0\id$, respectively.

\paragraph{Constructing the physical tangent space.}

Following Section~\ref{sec:TDVP}, let
$T_{\bar\psi}\subset\bar\psi^\perp$ denote the Hilbert-space
representatives of the projective tangent space at $[\bar\psi]$.
In practice, we construct this space from the horizontal parameter subspace $H_{\bar p}$.
Using real and imaginary parts of the parameter variations,
we construct a real matrix $R$ whose nullspace encodes the
tangent constraints and the horizontal conditions of
Proposition~\ref{prop:horizontal_conditions}.
A basis of $\ker R$ then gives a real basis
$\{h_1,\dots,h_r\}$ of $H_{\bar p}$, where $r=\dim_{\mathbb R}\ker R$.

The submersion property established above implies that the
restriction of $D_{\bar p}\Phi_{\boldsymbol D,\boldsymbol d}$
to $H_{\bar p}$ is a real-linear isomorphism onto the physical
tangent space.
Using the contraction derivative $w$ and the projector
$Q_{\bar\psi}$ defined in Section~\ref{sec:TDVP}, we obtain
\[
\Theta_a
:=
Q_{\bar\psi}w(h_a)\in T_{\bar\psi},
\qquad a=1,\dots,r.
\]
These vectors represent
$D_{\bar p}\Phi_{\boldsymbol D,\boldsymbol d}(h_a)$
and form a real basis of $T_{\bar\psi}$.
Collecting them as columns gives
\begin{equation}\label{eq:phys_tangent_matrix}
F=[\Theta_1\ \cdots\ \Theta_r]
\in\mathbb C^{(\prod_{i=1}^N d_i)\times r}.
\end{equation}
The physical tangent space is already a complex vector space.
Consequently,
\[
\operatorname{span}_{\mathbb R}\{\Theta_a\}
=
\operatorname{span}_{\mathbb C}\{\Theta_a\}
=
T_{\bar\psi}.
\]
In particular, if
$s:=\dim_{\mathbb C}T_{\bar\psi}$, then $r=2s$ and
\[
\operatorname{Im}_{\mathbb C}F=T_{\bar\psi},
\qquad
\rank_{\mathbb C}F=s.
\]
Thus, there is a redundancy among the $r$ columns as complex generators.
We obtain a complex orthonormal basis of the physical tangent space from the reduced singular
value decomposition
\[
F=U\Sigma V_s^\dagger,
\qquad
U^\dagger U=\id_s,
\qquad
\operatorname{Im}_{\mathbb C}U=T_{\bar\psi},
\]
where $\Sigma$ contains the $s$ nonzero singular values.
Retaining all nonzero singular values changes the basis,
not the tangent space. Numerically, nonzero singular values are identified using
the criterion specified in Appendix~\ref{sec:numerical_details}; discarding
additional nonzero values introduces a further approximation.

The matrices $F$ and $U$ need not be formed explicitly.
The same reduction can be implemented through MPS/MPO
contractions of their inner products, as described in
Appendix~\ref{sec:comp_costs}.

\paragraph{Linear ground-state perturbations.}

Every time-dependent tangent fluctuation now has the unique
representation
\[
\delta\psi(t)=U\eta(t),
\qquad
\eta(t)\in\mathbb C^s, 
\]
where $\eta(t)$ are time-dependent coefficients of the tangent space and $U$ is evaluated at $\bar\psi$ and is independent of time.

For a small real parameter $\varepsilon$, consider
\[
\psi_\varepsilon(t)
=
e^{-i\omega_0t}
\bigl(\bar\psi+\varepsilon U\eta(t)\bigr).
\]
The candidate vectors $\bar\psi+\varepsilon U\eta$ lie in the
affine tangent plane through $\bar\psi$; they need not be MPS
of the original bond dimensions.
The phase factor makes the fluctuation frequencies relative
to the variational reference energy $\omega_0$.

Substituting $\psi_\varepsilon(t)$ into \eqref{eq:TDVP_lagrangian}, the zeroth-order
term vanishes by the definition of $\omega_0$, and the
first-order term vanishes since
\[
U^\dagger H'\bar\psi=0,
\qquad
U^\dagger\bar\psi=0.
\]
The second-order term reads
\[
\mathcal L^{(2)}(\eta,\dot\eta)
=
\frac{i}{2}
\bigl(
\eta^\dagger\dot\eta-\dot\eta^\dagger\eta
\bigr)
-
\eta^\dagger B_{\mathrm{red}}\eta,
\]
where $B_{\mathrm{red}}
:=
U^\dagger H'U
=
B_{\mathrm{red}}^\dagger$.

Treating $\eta$ and its complex conjugate $\bar\eta$ as
independent variables, the Euler--Lagrange equations of this
quadratic Lagrangian,
\[
\frac{\partial\mathcal L^{(2)}}{\partial\bar\eta}
-
\frac{d}{dt}
\frac{\partial\mathcal L^{(2)}}{\partial\dot{\bar\eta}}
=0,
\]
give $i\dot\eta=B_{\mathrm{red}}\eta$.
The conjugate equations follow by complex conjugation.

\paragraph{Effective eigenvalue problem.}

For a harmonic ansatz $\eta(t)=y\,e^{-i\omega t}$,
the equations of motion become
\begin{equation}\label{eq:red_eigval_problem}
B_{\mathrm{red}}y=\omega y.
\end{equation}
The corresponding Hilbert-space excitation vector is
\[
\varphi:=Uy\in T_{\bar\psi}.
\]
Equivalently, the eigenproblem reads
\[
P_{T_{\bar\psi}}H'\varphi=\omega\varphi,
\qquad
P_{T_{\bar\psi}}:=UU^\dagger.
\]
Thus we diagonalize the Hamiltonian projected onto the
fixed physical tangent space.

For a normalized eigenvector $y$, the tangent vector $\varphi$
is also normalized, and $\langle\varphi,H\varphi\rangle=\omega_0+\omega$.
The eigenvalues $\omega$ therefore provide tangent-space
Rayleigh--Ritz estimates of excitation energies measured
relative to the variational reference energy $\omega_0$. Equivalently, the eigenvectors $y$ are the excitation states within the tangent approximation. Although we only consider the eigenvalues in this work, the corresponding eigenvectors represent another valuable resource for computations. 

This construction uses a fixed linear tangent ansatz.
A full linearization of TDVP on the MPS manifold would
additionally include second derivatives of the MPS
parametrization coupled to the residual $H'\bar\psi$.

In Appendix~\ref{sec:comp_costs}, we describe the computational
costs of constructing and diagonalizing the effective Hamiltonian.

\section{Bose--Hubbard benchmarks}\label{sec:results}

We illustrate the tangent-space method using finite
Bose--Hubbard (BH) chains. The model provides both particle-number
conserving and symmetry-breaking examples within the same setting,
allowing us to test the reconstructed spectra against exact
diagonalization. We stress that the outlined tangent-space construction is not specific to this model and the same projected eigenproblem applies to other Hamiltonians.
Further details on the numerical implementation, convergence
criteria, and rank tolerances are given in
Appendix~\ref{sec:numerical_details}.

\paragraph{Model and numerical setting.}

We consider an open chain with nearest-neighbour hopping,
on-site interactions, and a particle-number-breaking source term,
\begin{align}
H
&=
-J\sum_{i=1}^{N-1}
\left(
a_i^\dagger a_{i+1}+\mathrm{h.c.}
\right)
+\frac{U}{2}\sum_{i=1}^N n_i(n_i-1)
\nonumber\\
&\qquad
-\mu\sum_{i=1}^N n_i
+\eta\sum_{i=1}^N(a_i^\dagger+a_i),
\label{eq:bose_hubbard}
\end{align}
where $n_i=a_i^\dagger a_i$, $J$ is the hopping amplitude,
$U$ is the on-site interaction strength, $\mu$ is the chemical
potential, and $\eta$ is a real source strength.
For $\eta=0$, the Hamiltonian conserves the total particle number
\[
\hat M:=\sum_{i=1}^N n_i.
\]
A nonzero source $\eta$ breaks this $\mathrm U(1)$ symmetry.

We choose $N=10$, so that exact diagonalization remains feasible
and provides a benchmark for the tangent-space calculation.
Both calculations use the same local truncation to occupations
$n_i=0,1,2$, corresponding to $d=3$.
We set $U=1$, choose $\mu=0.5$, and vary $J\in[0.01,0.3]$.

For the spectra in Figure~\ref{fig:N_10_D_8_d_3_mu_0.5} (a) and (c),
the internal bond dimensions are
\[
(D_1,\dots,D_9)=(3,8,8,8,8,8,8,8,3).
\]

The reconstructed excitation energies $\omega$ are the
eigenvalues of $U^\dagger(H-\omega_0\id)U$, as defined in
Section~\ref{sec:linear_tangent_mps}.
We compare them with excitation energies from exact diagonalization of the BH Hamiltonian.

\paragraph{Number-conserving spectrum.}

For $\eta=0$, we choose reference states with definite particle
number. We denote the reference-state particle number by $M_0$
and the particle number of an excitation state by $M$.
Figure~\ref{fig:N_10_D_8_d_3_mu_0.5} (a) compares the exact and
reconstructed spectra. The marker colours and background shading encode different
information: marker colours
label the excitation sector $M$. The background shading
distinguishes particle-resolved Schmidt-rank profiles of the
reference state: $r_\ell(m)$ counts the nonzero Schmidt
coefficients associated with $m$ particles to the left of
cut $\ell$. These profiles are defined and analyzed in
Section~\ref{sec:schmidt_rank_distribution}.
The background colours therefore label rank profiles of the variational ground-state tensor.
Figure~\ref{fig:N_10_D_8_d_3_mu_0.5} (b) characterizes the one-body properties
of the two sector branches ($M=10,\,11$). For each normalized wave function of the ground state
$\chi_M$, we form the one-body density matrix and plot its largest eigenvalue $\lambda_0$ divided by $N$.
This quantity measures the occupation per site of the most
populated one-body orbital.

\begin{figure}[h!]
    \centering

    \begin{tikzpicture}
        \node[inner sep=0] (img1) {
            \includegraphics[width=0.49\textwidth]
            {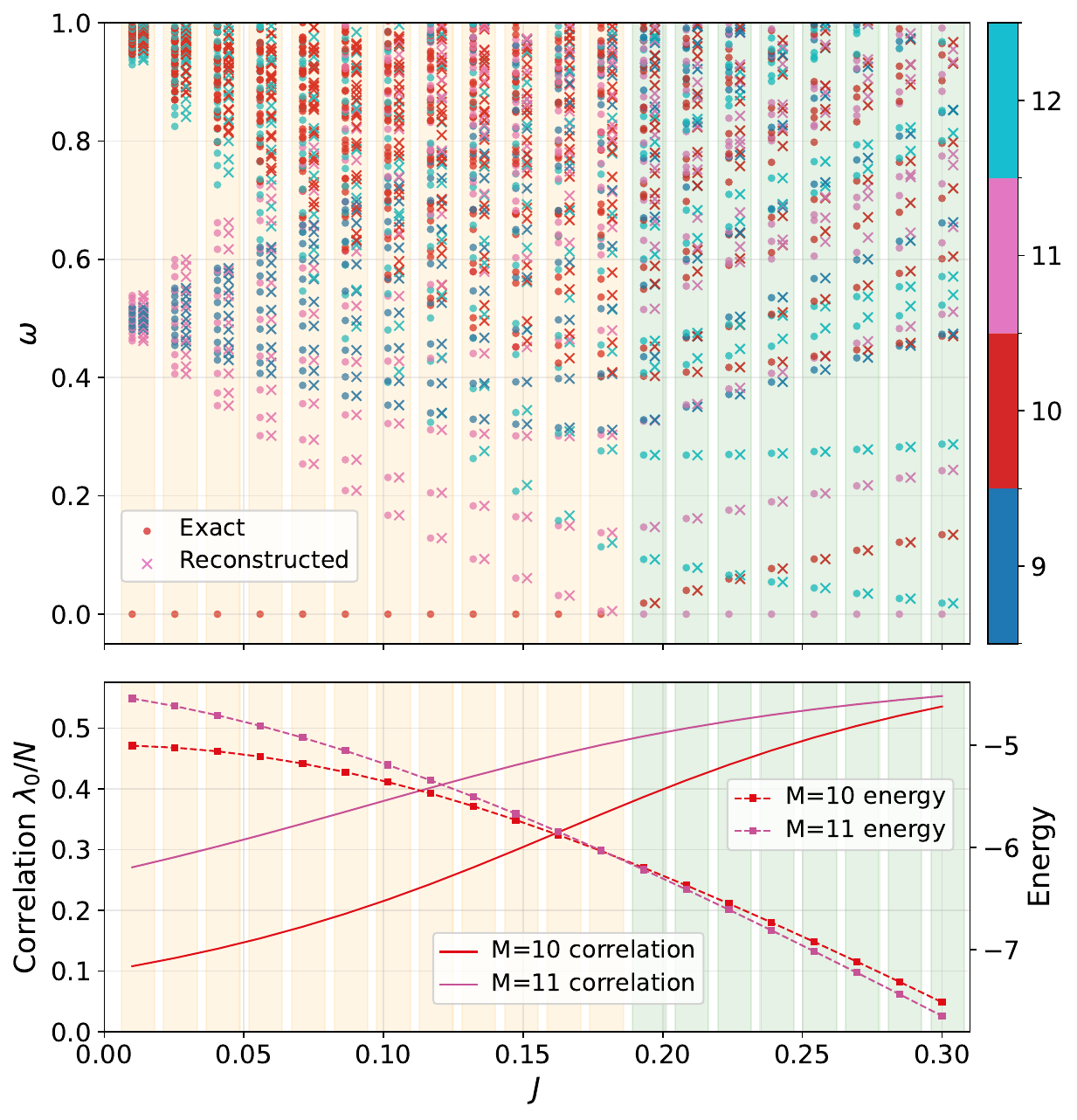}
        };

        \node[
            anchor=north west,
            font=\small\bfseries,
            inner sep=1pt
        ] at ([xshift=242pt,yshift=-3pt]img1.north west) {(a)};

        \node[
            anchor=north west,
            font=\small\bfseries,
            inner sep=1pt
        ] at ([xshift=242pt,yshift=-20pt]
              $(img1.north west)!0.52!(img1.south west)$) {(b)};
    \end{tikzpicture}

    \begin{tikzpicture}
        \node[inner sep=0] (img2) {
            \includegraphics[width=0.49\textwidth]
            {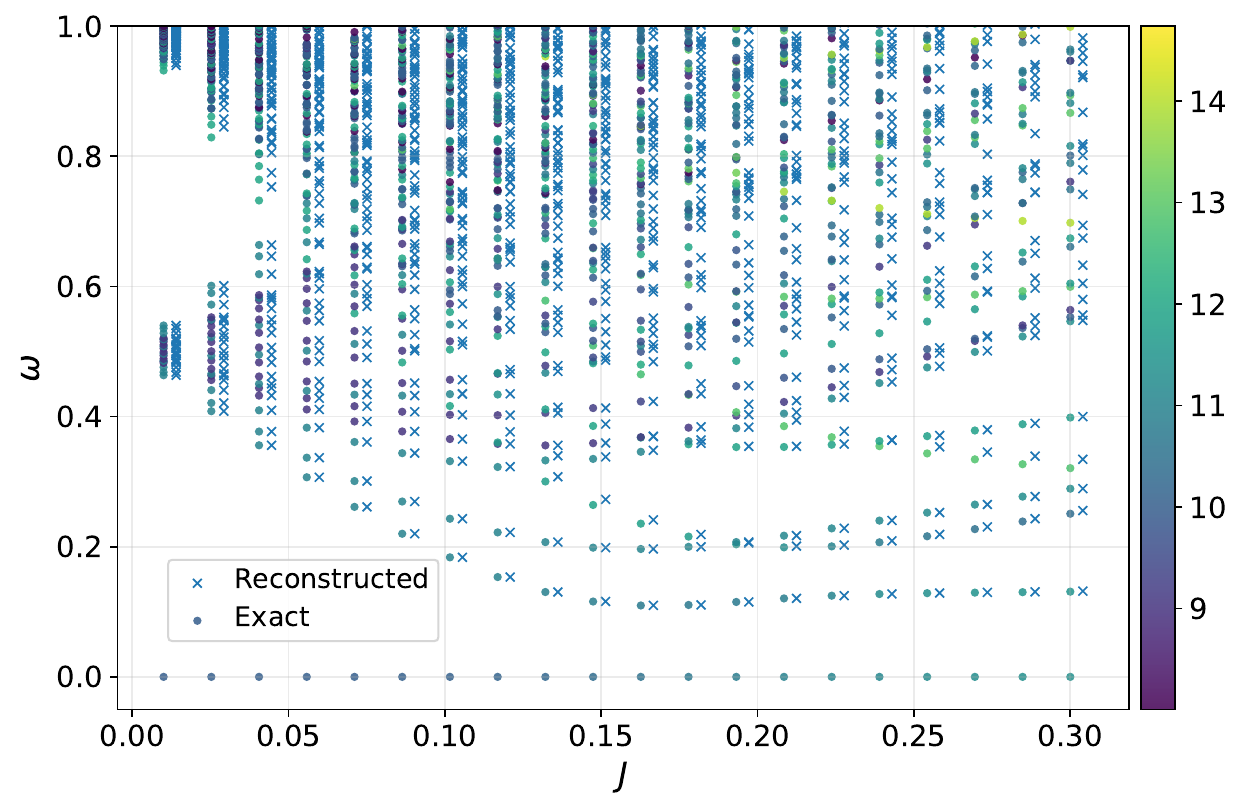}
        };

        \node[
            anchor=north west,
            font=\small\bfseries,
            inner sep=1pt
        ] at ([xshift=242pt,yshift=-3pt]img2.north west) {(c)};
    \end{tikzpicture}

    \caption{
    Bose--Hubbard benchmarks for $N=10$, maximum bond dimension
    $D=8$, local dimension $d=3$, and $\mu/U=0.5$.
    (a) Exact excitation energies (dots) and tangent-space
    estimates (crosses) for $\eta=0$; marker colours label
    the particle-number sectors $M$.
    (b) Energies of the lowest states in sectors $M=10$
    and $M=11$ (right axis), together with
    $\lambda_0/N$ for their one-body density matrices
    (left axis). Background shading in panels (a) and (b)
    distinguishes rank profiles of the reference state,
    independently of the marker colours; see
    Section~\ref{sec:schmidt_rank_distribution}.
    In this example the colouring coincides with ground-state
    occupation $M_0$.
    (c) Exact and reconstructed excitation energies for
    $\eta/U=0.01$; marker colours indicate
    $\langle\hat M\rangle$.
    }
    \label{fig:N_10_D_8_d_3_mu_0.5}
\end{figure}

In the orange-shaded region of (a), the reference
state has $M_0=10$. The lowest particle-addition excitation lies
in sector $M=11$ and softens as $J$ increases.
Its excitation energy vanishes near $J_c\simeq0.18$, where
the ground state changes from sector $M_0=10$ to sector
$M_0=11$.
For this finite, locally truncated chain, the gap closing is a
ground-state sector crossing. We interpret it as a finite-size
precursor of the Mott-insulator--superfluid boundary, rather
than as a determination of the thermodynamic transition
\cite{Fisher_1989}.
The simultaneous growth of $\lambda_0/N$ in (b)
is consistent with increasing one-body coherence.

After the crossing, the reference state has $M_0=11$, in the
green-shaded region shown here. The lowest particle-removal
and particle-addition excitations then lie in sectors
$M=10$ and $M=12$, respectively.
The reconstruction of the low-lying $M=12$ branch improves
after this change of reference state: before the crossing, this sector corresponds to adding two
particles to the reference; afterwards, it corresponds to
adding one.
The tangent trial space also changes at the crossing.
The sector-resolved analysis in
Section~\ref{sec:tangent_space_tomography} quantifies the
associated changes in expressivity. 

\paragraph{Accuracy depends on bond dimension.}

Figure~\ref{fig:approx_error} shows the absolute error
$|\omega-\omega_{\mathrm{exact}}|$ for the lowest
particle-removal and particle-addition modes, in sectors
$M=M_0-1$ and $M=M_0+1$, respectively, for several bond dimensions.
The errors are small at weak hopping and generally increase
away from this regime. Increasing the bond dimension improves
the reconstruction over much of the parameter range, although
the improvement depends on the excitation branch and on $J$.
These results illustrate the increasing demands on a
fixed-bond-dimension ansatz as the reference state and its
excitations depart from the weak-hopping, weak-correlation regime.

The error curves also change abruptly near the ground-state
sector crossing. There the reference state, and hence its
tangent space, switches between different rank profiles.
Moreover, following $M=M_0\pm1$ changes the absolute
particle-number sectors being compared.
The jump should therefore be interpreted together with
the change of reference state, its PRSR profile, and the
excitation sectors being followed.
Section~\ref{sec:schmidt_rank_distribution} analyzes how the
reference-state rank profile controls the available
tangent-sector dimensions.

\begin{figure}[h]
    \centering
    \includegraphics[width=0.48\textwidth]
    {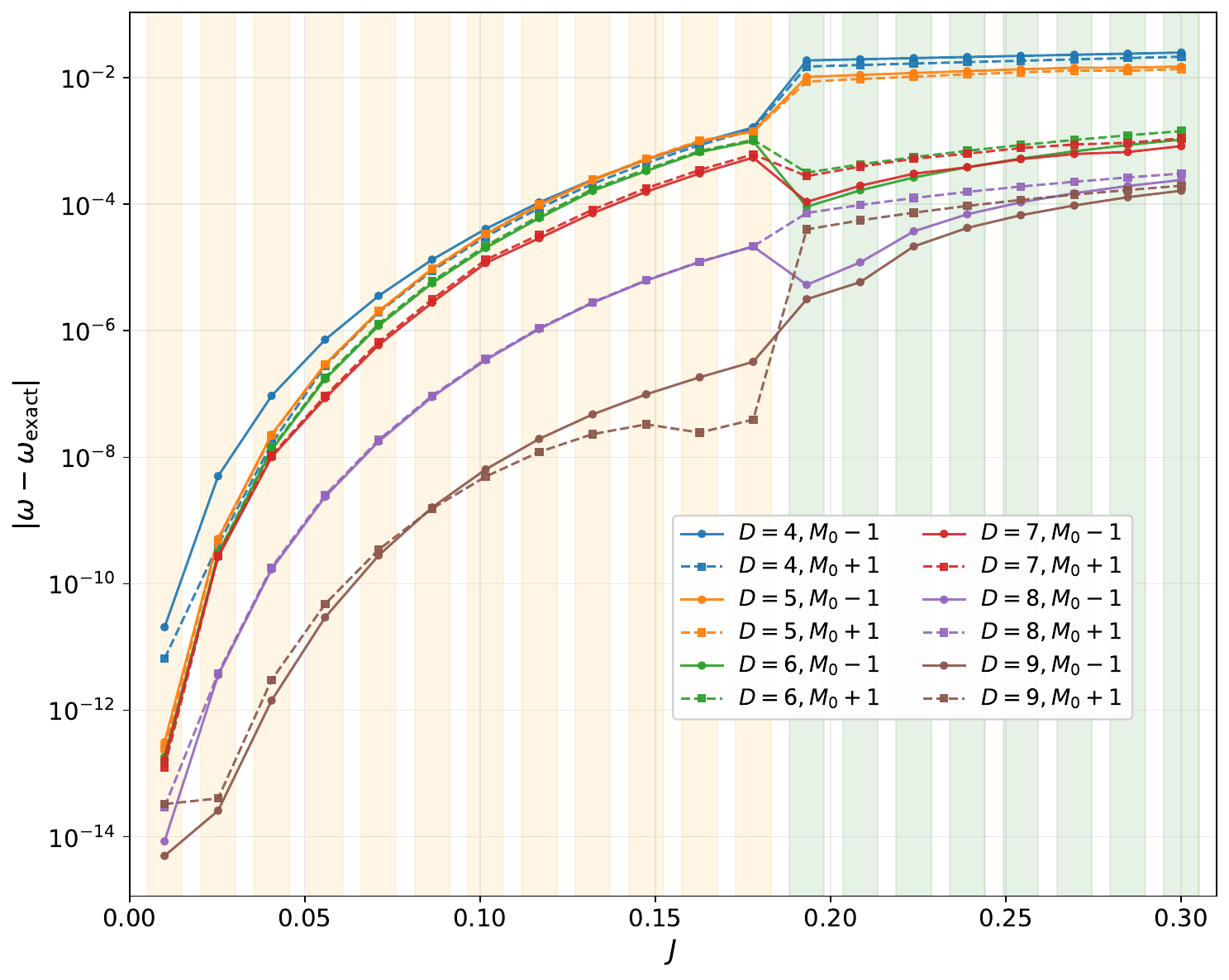}
    \caption{
    Absolute excitation-energy errors
    $|\omega-\omega_{\mathrm{exact}}|$ for the lowest
    particle-removal and particle-addition modes of the
    number-conserving Bose--Hubbard chain.
    Here $N=10$, $d=3$, $\mu/U=0.5$, and the maximum bond
    dimension varies from $D=4$ to $D=9$.
    The modes are followed in sectors $M=M_0-1$ and
    $M=M_0+1$, respectively, where $M_0$ is the
    reference-state particle number.
    Background shading follows the region labels in
    Figure~\ref{fig:N_10_D_8_d_3_mu_0.5}.
    }
    \label{fig:approx_error}
\end{figure}

\paragraph{Broken particle-number symmetry.}

The tangent-space construction does not require particle-number
conservation. We demonstrate this using $\eta=0.01$, shown in
Figure~\ref{fig:N_10_D_8_d_3_mu_0.5} (c).
Since $[H,\hat M]\neq0$, particle number is no longer a
conserved label. We therefore colour each exact
state by its expectation value $\langle\hat M\rangle$,
rather than assigning an integer-valued sector.

In this example, the ground-state expectation value
$\langle\hat M\rangle_0$ varies continuously with $J$, in
contrast to the sector crossing at $\eta=0$.
The reconstructed low-energy branches remain in close
agreement with the exact spectrum.
This tests the excitation construction beyond the
number-conserving setting.

Finally, Appendix~\ref{sec:exact_spectrum} provides an exact-limit
consistency check with $N=3$, $d=3$, internal bond dimensions
$(D_1,D_2)=(3,3)$, $\mu=0.5$, and $U=1$.
For these dimensions, the MPS variety is the full projective
Hilbert space. At an exact reference state, its physical
tangent space is $\bar\psi^\perp$ and contains all excitation
directions. The projected eigenproblem therefore reproduces
the complete excitation spectrum in exact arithmetic.
The numerical comparison checks the implementation in
this limit.

\section{Tangent-space tomography}
\label{sec:tangent_space_tomography}

The Bose--Hubbard benchmarks show that the accuracy of the
tangent-space spectrum can depend strongly on the particle-number
sector. We now examine the underlying representational question:
which particle-number subspaces are contained in the tangent ansatz,
and how many independent directions are missing?
We address this question through a sector-resolved
\emph{tangent-space tomography}. First, we define the parametric
deficiency as a measure of sector incompleteness. We then relate
the sector tangent dimensions to the particle-resolved Schmidt
ranks of the reference state.
This distinguishes structural restrictions of the trial space
from the accuracy of individual excitation energies. Note that for such a sector-resolved analysis we require that particle number is a well-defined label, that is, $[H, \hat M] = 0$.

\subsection{Parametric deficiency}

Let $\bar\psi$ be the normalized full-rank reference state used
in Section~\ref{sec:linear_tangent_mps}.
Recall that the columns of $U$ form a complex orthonormal basis
of its physical tangent space,
\[
T_{\bar\psi}
=
\operatorname{Im}_{\mathbb C}U
\subseteq\bar\psi^\perp\subseteq\mathcal H.
\]
This is the Hilbert-space representation of
$T_{[\bar\psi]}V^{=}_{\boldsymbol D,\boldsymbol d}$.
Throughout this section, dimensions and ranks are taken over
$\mathbb C$.
The Hilbert space decomposes into eigenspaces of the total
particle-number operator:
\[
\mathcal H
=
\bigoplus_{M=0}^{M_{\max}}\mathcal H_M,
\qquad
M_{\max}:=\sum_{i=1}^N(d_i-1).
\]
Let $P_M:\mathcal H\longrightarrow\mathcal H_M$
denote the orthogonal projector onto sector $M$.
Ordering the occupation-number basis by total particle number
writes any state vector as
\[
\psi
=
\bigl(\psi^{(0)},\dots,\psi^{(M_{\max})}\bigr),
\qquad
\psi^{(M)}=P_M\psi.
\]
For homogeneous local dimensions $d_i=d$, the sector dimension is
\[
\dim_{\mathbb C}\mathcal H_M
=
\sum_{k=0}^{\lfloor M/d\rfloor}
(-1)^k
\binom Nk
\binom{M-kd+N-1}{N-1}.
\]
Similarly, each tangent vector decomposes as
\[
t=\sum_M P_Mt,
\qquad t\in T_{\bar\psi},
\]
and may have nonzero components in several particle-number sectors.

\paragraph{Projected and intrinsic sector spaces.}

To count the available directions sectorwise, we distinguish
the \emph{projected tangent space}
\[
T_{\bar\psi,M}^{\mathrm{proj}}
:=
P_MT_{\bar\psi}
=
\{P_Mt:t\in T_{\bar\psi}\}
\subseteq\mathcal H_M
\]
from the \emph{intrinsic sector tangent space}
\[
T_{\bar\psi,M}^{\mathrm{intr}}
:=
T_{\bar\psi}\cap\mathcal H_M.
\]
The former contains all sector-$M$ components of tangent vectors;
the latter contains tangent vectors lying entirely in sector $M$.
Since $P_Mt=t$ for $t\in\mathcal H_M$,
\[
T_{\bar\psi,M}^{\mathrm{intr}}
\subseteq
T_{\bar\psi,M}^{\mathrm{proj}}.
\]
Equality, however, does not hold for a general linear trial space.
For example, let
$\mathcal H=\mathcal H_0\oplus\mathcal H_1$,
choose nonzero $e_0\in\mathcal H_0$ and $e_1\in\mathcal H_1$,
and consider $\mathcal S=\operatorname{span}_{\mathbb C}\{e_0+e_1\}$.
Then
\[
P_0\mathcal S=\operatorname{span}_{\mathbb C}\{e_0\},
\qquad
\mathcal S\cap\mathcal H_0=\{0\}.
\]
Although $e_0$ occurs as a projected component, it is not an
allowed trial vector in $\mathcal S$: every nonzero vector in
$\mathcal S$ also has a sector-$1$ component.
Thus projection can produce vectors outside the original trial
space. Counting the projected dimension would then overestimate
the number of available directions lying entirely in sector $0$, which is $\dim_{\mathbb C}(\mathcal S\cap\mathcal H_0)$.

\paragraph{Particle-number-adapted tangent spaces.}

We call $T_{\bar\psi}$ \emph{adapted} to the particle-number
decomposition if
\[
P_MT_{\bar\psi}\subseteq T_{\bar\psi}
\qquad\text{for every }M.
\]
In that case, each particle-number component of a tangent vector
is itself a tangent vector, such that
\[
T_{\bar\psi}
=
\bigoplus_M
\bigl(T_{\bar\psi}\cap\mathcal H_M\bigr).
\]
and $T_{\bar\psi,M}^{\mathrm{proj}}
=
T_{\bar\psi,M}^{\mathrm{intr}}$.
We write
\[
T_{\bar\psi,M}
:=
P_MT_{\bar\psi}
=
T_{\bar\psi}\cap\mathcal H_M
\]
for the common projected and intrinsic sector space.
Its dimension is
\[
\rho_M
:=
\dim_{\mathbb C}T_{\bar\psi,M}
=
\rank_{\mathbb C}(P_MU).
\]
Thus $\rho_M$ counts the independent trial directions available
entirely within particle-number sector $M$. Note that if $\bar \psi$ was obtained from a Hamiltonian with $[H, \hat M] \neq 0$, then $T_{\bar \psi}$ might \emph{not} be adapted and $T_{\bar\psi,M}^{\mathrm{proj}}
\neq
T_{\bar\psi,M}^{\mathrm{intr}}$

\paragraph{Missing sector directions.}

Because the physical tangent space is orthogonal to $\bar\psi$,
the \emph{target space} in sector $M$ is
\[
W_M:=\mathcal H_M\cap\bar\psi^\perp.
\]
For $M\neq M_0$, orthogonality to $\bar\psi$ is inherited from $\mathcal H$.
For $M=M_0$, it removes the complex reference-state line.
Hence
\[
d_M^{\mathrm{target}}
:=
\dim_{\mathbb C}W_M
=
\begin{cases}
\dim_{\mathbb C}\mathcal H_M-1, & M=M_0,\\
\dim_{\mathbb C}\mathcal H_M,   & M\neq M_0.
\end{cases}
\]

Since $T_{\bar\psi,M}\subseteq W_M$, we define its
sector-resolved \emph{parametric deficiency} by
\[
\delta_M^{\mathrm{par}}
:=
d_M^{\mathrm{target}}-\rho_M
\geq0.
\]
This is the codimension of the sector tangent space in $W_M$:
it counts how many additional independent directions are
needed to span the target space.
In particular,
\[
\delta_M^{\mathrm{par}}=0
\quad\Longleftrightarrow\quad
T_{\bar\psi,M}=W_M.
\]
The vector $(\delta_M^{\mathrm{par}})_M$ records the
incompleteness of the tangent ansatz across particle-number
sectors; examples are given in
Section~\ref{sec:examples_tomography}.
A non-zero deficiency indicates that a sector is not represented
completely. It does not identify which excitation states are
missing or what the effect on approximation accuracy will be. 
\paragraph{Rank frustration.}

Parametric deficiency measures missing directions in the
tangent space of a chosen reference state.
A related structural restriction concerns the ranks that the
reference state can attain at all.
Fixed total particle number and finite local occupations limit
the compatible states on either side of each cut, and therefore
bound the Schmidt ranks. We refer to this restriction as
\emph{rank frustration}.
It is a possible source of parametric deficiency, rather than
an additional deficit to be added to
$\delta_M^{\mathrm{par}}$.

For example, for $N=5$ and local occupations $0,1,2$, the maximal
internal bond patterns in the indicated sectors are
\begin{equation}\label{eq:rank_frustration}
\boldsymbol D_{\mathrm{int}}^{\max}(M_0)
=
\begin{cases}
(3,9,9,3), & M_0=4,5,6,\\
(3,7,7,3), & M_0=3,7,\\
(3,4,4,3), & M_0=2,8.
\end{cases}
\end{equation}
Here the pattern is $(D_1,D_2, D_3,D_4)$, with
$D_0=D_5=1$.
Every state in the corresponding sector admits an MPS
representation with ranks bounded by the listed pattern.
Increasing a nominal bond dimension beyond these bounds cannot
increase the reference state's actual Schmidt ranks.
Nevertheless, its linear tangent space can remain deficient
in other particle-number sectors, even when the reference
state is represented exactly.

One possible extension is to enlarge the trial space beyond
first-order tangent vectors. A first-order variation replaces
one core at a time. Allowing simultaneous variations of two
cores, with independent coefficients in the excitation ansatz,
can supply joint variations absent from the first-order space.
Adding these double-tangent contributions might reduce
sector deficiencies, although it need not eliminate them
\cite{Haegeman_2013}.

\subsection{Particle-resolved Schmidt rank profile}
\label{sec:schmidt_rank_distribution}

The sector ranks $\rho_M$ count the independent excitation directions
available in the physical tangent space. To understand their dependence
on the reference state, we resolve its Schmidt ranks by particle number.
The resulting \emph{particle-resolved Schmidt-rank profile} (PRSR)
records how the bond dimension across each cut is distributed among
particle-number sectors of the left subsystem.

Symmetry-resolved bond ranks and the associated block structure
have been investigated in \cite{Bachmayr_2022, Zauner_2018}. Here we derive a similar structure to compute the
particle-sector dimensions of the physical tangent space for finite, non-uniform systems.
The resulting formula allows $\delta^{\rm par}_M$ to be diagnosed
from the reference-state ranks without constructing a tangent basis.

\paragraph{Particle-number blocks.}

Fix a normalized full-rank reference state
$\bar\psi\in\mathcal H_{M_0}$, and consider the cut after site $\ell$.
If the left subsystem contains $m$ particles, the right subsystem
must contain $M_0-m$ particles. Consequently, the coefficient
matrix separates into blocks,
\begin{equation}\label{eq:mat_form_GS}
\mathrm{Mat}_\ell(\bar\psi)
\cong
\bigoplus_m\Psi_\ell^{(m)}.
\end{equation}
The block $\Psi_\ell^{(m)}$ contains the coefficients of the wave-function
component in
\[
\mathcal H^{[1:\ell]}_m
\otimes
\mathcal H^{[\ell+1:N]}_{M_0-m}.
\]
Here orthonormal row and column bases are grouped into complementary
particle-number sectors.
An explicit coefficient-matrix example is given in
Appendix~\ref{app:coeff_matrix}.

\paragraph{Particle-resolved Schmidt ranks.}

We define the PRSR as
\begin{equation}
    r_\ell(m):=\rank_{\mathbb C}\Psi_\ell^{(m)}.
\end{equation}
The singular values of $\Psi_\ell^{(m)}$ are the Schmidt coefficients
associated with $m$ particles on the left of the cut.
Thus $r_\ell(m)$ counts the nonzero Schmidt coefficients in that
sector, and $r_\ell=(r_\ell(m))_m$
is the PRSR profile at cut $\ell$.
We set $r_\ell(m)=0$ whenever that sector is absent.

The blocks in \eqref{eq:mat_form_GS} occupy separate row and column
subspaces, so their ranks add.
Since the reference state is full-rank, we can resolve the bond dimension as
\begin{equation}\label{eq:PRSR_bond_dim}
D_\ell
=
\rank\bigl(\mathrm{Mat}_\ell(\bar\psi)\bigr)
=
\sum_m r_\ell(m).
\end{equation}
Thus $D_\ell$ gives the total Schmidt rank, whereas $r_\ell(m)$
specifies its contribution from each particle-number sector.
Two states can have identical bond dimensions but different
PRSR profiles.

\paragraph{Block structure of reference-state matrices.}

To relate this rank profile to the tensors that parametrize the reference state, first consider
what a bond index represents. Contracting the first $\ell$
MPS cores while fixing the remaining bond index produces a
state of sites $1,\dots,\ell$.
The $D_\ell$ bond-index values therefore label $D_\ell$ states
of this entire left subsystem.
In a full-rank left-canonical representation, these states
are orthonormal and span the space of left Schmidt states
with non-zero Schmidt coefficients.
We can therefore choose these subsystem states to have
definite particle number: exactly $r_\ell(m)$ of them contain
$m$ particles.
Grouping the bond-index values by this number allows us to write
\[
\mathbb C^{D_\ell}
\cong
\bigoplus_m\mathbb C^{r_\ell(m)}.
\]
This grouping may require a unitary change of bond coordinates but it does not change the represented state.
In particular, an MPS obtained from DMRG need not display the
blocks in its original coordinates, even when the represented
state has definite particle number.

Now consider $M^j_n\in\mathbb C^{D_{j-1}\times D_j}$.
Its row index labels a state of sites $1,\dots,j-1$,
while its column index labels a state of sites $1,\dots,j$.
Its entries are the coefficients used to express the latter
state by combining a state of the first $j-1$ sites with
occupation $n$ at site $j$.

The rows can consequently be grouped by the particle number
$m$ before adding site $j$, and the columns by the particle
number $m'$ after adding it:
\begin{align}\label{eq:structure_MPS_matrices}
M^j_n:
\bigoplus_{m'}\mathbb C^{r_j(m')}
&\longrightarrow
\bigoplus_m\mathbb C^{r_{j-1}(m)}, \nonumber\\
M^j_n(m,m'):
\mathbb C^{r_j(m')}
&\longrightarrow
\mathbb C^{r_{j-1}(m)}.
\end{align}
Thus $M^j_n(m,m')$ is a $r_{j-1}(m)\times r_j(m')$ block.

Particle counting determines which blocks can be nonzero.
A row state with $m$ particles, combined with occupation $n$
at site $j$, has $m+n$ particles.
It cannot contribute to a column state with a different
particle number. Hence
\[
M^j_n(m,m')=0
\qquad\text{unless}\qquad
m'=m+n.
\]
For example, if $r_{j-1}(1)=2$ and $r_j(2)=3$, then $M^j_1$
can contain a $2\times3$ block connecting these groups:
one particle in the left subsystem and one at site $j$
give two particles in the enlarged subsystem.
The same row and column groups in $M^j_0$ must have zero
entries, because adding an empty site cannot increase
the particle number from one to two.
For fixed $n$, the only possibly nonzero blocks are therefore
\[
M^j_n(m):=M^j_n(m,m+n):
\mathbb C^{r_j(m+n)}
\longrightarrow
\mathbb C^{r_{j-1}(m)}.
\]
After pairing the corresponding row and column groups and
suppressing zero rows or columns, this gives
\[
M^j_n\cong\bigoplus_m M^j_n(m).
\]
Therefore, the block structure of the GS matrices explicitly encodes particle-number conservation of the physical system. 
An explicit example is given in
Appendix~\ref{app:block_form_matrices}.

\paragraph{Block structure of perturbations.}

The same row and column groups organize the first-order
variations $\delta M^j_n$.
However, the tangent space used for excitations is that of
the full MPS manifold, not only its fixed-$M_0$ subset.
We therefore allow variations that do not satisfy the
reference-state selection rule $m'=m+n$.

A variation block
\[
\delta M^j_n(m,m'):
\mathbb C^{r_j(m')}
\longrightarrow
\mathbb C^{r_{j-1}(m)}
\]
has entries to vary only when
\[
r_{j-1}(m)>0
\qquad\text{and}\qquad
r_j(m')>0.
\]
Thus a block that vanishes in the reference tensor may
acquire a nonzero variation, provided both index groups
are present. A block with no rows or no columns cannot.

The particle-number sector of an insertion follows from
the unchanged tensors on either side.
The left subsystem contains $m$ particles, site $j$
contributes $n$, and the right subsystem contains
$M_0-m'$ particles. Hence the resulting contribution has
total particle number
\begin{equation}\label{eq:tangent_sector_shift}
M=M_0+m+n-m'.
\end{equation}
Blocks satisfying $m'=m+n$ preserve $M_0$.
Other blocks can produce particle-addition or particle-removal
components.

\paragraph{From PRSR to parametric deficiency.}
We can count the independent physical tangent directions in each
particle-number sector via the PRSR profile. For this purpose, we use the affine MPS
parametrization, where core variations have independent complex
entries. This avoids counting directly with the coupled real
Stiefel and horizontal constraints.

We keep the same full-rank reference tensors and choose a normalized
representative $C$ of the final projective core. We regard the tuple
\[
A=(A^1,\dots,A^N)\in\mathcal A_{\boldsymbol D,\boldsymbol d},
\]
with $~A^i = M^i$ and $A^N = C$, as a point of the affine parameter space, and allow unrestricted
complex core variations $B=(B^1,\dots,B^N)$.
The reference tensors remain left canonical, but the variations
$B$ are not required to preserve this canonical form.

The affine and left-canonical parametrizations describe the same full-rank MPS states. At a full-rank reference, the affine contraction differential therefore parametrizes the tangent space of the affine MPS cone. Its kernel consists of gauge variations, while projectivization removes the additional scaling direction \(\mathbb C\bar\psi\).
Writing
\[
\mathcal L_A(B)
:=
D_A\widetilde\Phi_{\boldsymbol D,\boldsymbol d}(B)
\]
for the affine contraction differential, its image satisfies
\[
\operatorname{Im}\mathcal L_A
=
\mathbb C\bar\psi\oplus T_{\bar\psi}.
\]
Thus the affine differential generates the same physical tangent
directions as the horizontal Stiefel construction, together with
the reference-state line, because the vertical directions map to zero.
In the following, we compute the dimension of the image of the affine differential by subtracting the dimension of its kernel from the number of core-variation parameters, and then remove
the reference-state line to obtain $\rho_M$.

Introduce the particle-number shift
$\Delta:=M-M_0$ and the boundary profiles
\[
r_0(m):=\delta_{m,0},
\qquad
r_N(m):=\delta_{m,M_0}.
\]
All absent particle-number groups are assigned rank zero.
The sums over $m\in\mathbb Z$ below are therefore finite.

By \eqref{eq:tangent_sector_shift}, a block
$B^j_n(m,m')$ contributes to sector $M_0+\Delta$ when $m'=m+n-\Delta$.
It contains $r_{j-1}(m)r_j(m+n-\Delta)$ independent complex
entries. Let $\mathcal B_\Delta$ denote the space of core
variations supported on these blocks. Then
\[
\mathcal L_A(\mathcal B_\Delta)
\subseteq\mathcal H_{M_0+\Delta},
\]
and its parameter dimension is
\begin{align}\label{eq:PRSR_site_count}
S_\Delta
&:=
\dim_{\mathbb C}\mathcal B_\Delta \nonumber\\
&=
\sum_{j=1}^{N}
\sum_{n=0}^{d_j-1}
\sum_{m\in\mathbb Z}
r_{j-1}(m)\,r_j(m+n-\Delta).
\end{align}
This counts core variations, not yet independent physical
tangent vectors. To identify the redundancy, choose a matrix
$X_\ell\in\mathbb C^{D_\ell\times D_\ell}$ on an internal
bond and set
\[
B^\ell_n=A^\ell_nX_\ell,
\qquad
B^{\ell+1}_{n'}=-X_\ell A^{\ell+1}_{n'},
\]
with all other core variations zero.
The two terms in the contraction differential insert the
same matrix at the same bond with opposite signs:
\[
(A^\ell_nX_\ell)A^{\ell+1}_{n'}
+
A^\ell_n(-X_\ell A^{\ell+1}_{n'})
=0.
\]
Hence $\mathcal L_A(B)=0$ for every choice of $X_\ell$.

For all internal bonds simultaneously, the gauge variations are
\[
B^j_n=A^j_nX_j-X_{j-1}A^j_n,
\qquad
X_0=X_N=0.
\]
At a full-rank reference, these variations constitute the entire
kernel of $\mathcal L_A$, and the tuple $(X_1,\dots,X_{N-1})$
parametrizes that kernel injectively \cite{Haegeman_2014}.
The formula has the same local form as
\eqref{eq:vertical_directions_explicit}, but here the $X_j$
are arbitrary complex matrices rather than skew-Hermitian
matrices: the affine variations here are allowed to leave the
left-canonical parameter space.

To count this kernel sectorwise, consider a block
\[
X_\ell(m,m'):
\mathbb C^{r_\ell(m')}
\longrightarrow
\mathbb C^{r_\ell(m)}.
\]
Each of these blocks connects a left subsystem
with $m$ particles to a right subsystem with $M_0-m'$ particles.
Before the contributions are added, each therefore lies in
sector
\[
M=m+(M_0-m')=M_0+m-m'.
\]

For a fixed shift $\Delta$, the allowed gauge blocks therefore have
$m'=m-\Delta$ and contain
$r_\ell(m)r_\ell(m-\Delta)$ independent complex parameters.
Since the gauge parametrization is injective, summing these
block dimensions gives
\begin{equation*}
G_\Delta
:=
\dim_{\mathbb C}
\ker\!\left(\mathcal L_A|_{\mathcal B_\Delta}\right)
=
\sum_{\ell=1}^{N-1}
\sum_{m\in\mathbb Z}
r_\ell(m)\,r_\ell(m-\Delta).
\end{equation*}

\begin{proposition}\label{prop:PRSR_tangent_ranks}
Let $\bar\psi\in\mathcal H_{M_0}$ be a normalized full-rank MPS reference state, and let
$T_{\bar\psi}\subset\bar\psi^\perp$ be its physical
tangent space. Then
\begin{equation}\label{eq:PRSR_tangent_rank}
\rho_M=\rho_{M_0+\Delta}
=
S_\Delta-G_\Delta-\delta_{\Delta,0},
\end{equation}
where $\Delta = M - M_0$.
\end{proposition}

\begin{proof}
The core-variation space decomposes into the spaces
$\mathcal B_\Delta$, and $\mathcal L_A$ maps different shifts
into different particle-number sectors. Consequently,
\[
\operatorname{Im}\mathcal L_A
\cap\mathcal H_{M_0+\Delta}
=
\mathcal L_A(\mathcal B_\Delta).
\]
The domain of this restricted differential has dimension
$S_\Delta$, and its kernel has dimension $G_\Delta$.
Rank--nullity therefore gives
\[
\dim_{\mathbb C}
\bigl(
\operatorname{Im}\mathcal L_A
\cap\mathcal H_{M_0+\Delta}
\bigr)
=
S_\Delta-G_\Delta.
\]

Passing from the affine tangent space to $T_{\bar\psi}$
removes the complex line $\mathbb C\bar\psi$.
This line lies entirely in sector $M_0$, so exactly one
complex dimension is removed for $\Delta=0$ and none
for the other shifts. This proves
\eqref{eq:PRSR_tangent_rank}.
\end{proof}

\paragraph{Practical use of the diagnostic.}

Once the PRSR profile is known, the sector dimensions follow
from sums and products of its integer entries.
In particular, using
$d_M^{\mathrm{target}}
=\dim_{\mathbb C}\mathcal H_M-\delta_{M,M_0}$,
the parametric deficiency becomes
\begin{equation}\label{eq:PRSR_deficiency}
\delta_M^{\mathrm{par}}
=
\dim_{\mathbb C}\mathcal H_M
-
S_{\Delta}
+
G_{\Delta}.
\end{equation}
Neither a tangent-space basis nor the projected Hamiltonian
is needed to evaluate this quantity.
It identifies sectors that are absent, partially represented,
or complete before the excitation eigenproblem is solved.

As we illustrate in the following section, the PRSR profile of a ground state tensor might not be the most suitable for reconstructing excitations in a sector of interest , say $M^*$. Enforcing a more suitable PRSR profile in the ground state routine (DMRG), that is, choosing a profile with lower parametric deficiency in $M^*$, will produce a tangent space with better expressivity in $M^*$. However, this may increase the variational reference energy, since we force a rank structure that the energy-minimizing ground-state routine (see Appendix~\ref{sec:numerical_details}) otherwise would not choose. This introduces an approximation error of the excitation energies because the reference energy is more inaccurate. Therefore, knowledge of the PRSR profile allows the search for more suitable reference states within the tangent method. 

\subsection{Illustrative example}\label{sec:examples_tomography}

We illustrate how the PRSR determines sectorwise expressivity
at fixed bond dimensions using the number-conserving
Bose--Hubbard Hamiltonian \eqref{eq:bose_hubbard}.
We choose $N=4$, local occupations $n=0,1,2$ ($d=3$),
$U=1$, $\mu=0.5$, $\eta=0$, and vary
$J\in[0.01,0.3]$.
At each parameter value $J$, the variational reference state has a definite particle number $M_0$, which changes from sector $M_0=4$ to sector $M_0=5$ along the sweep.

The internal bond dimensions remain
\[
(D_1,D_2,D_3)=(3,5,3),
\qquad D_0=D_4=1
\]
throughout the parameter sweep.
Nevertheless, the reference states exhibit four different
PRSR profiles.
By \eqref{eq:PRSR_tangent_rank}, these profiles determine
different sector tangent dimensions $\rho_M$, even though
the coarse bond dimensions are unchanged.

\paragraph{Rank profiles and reconstructed spectra.}

\begin{figure}
    \centering

    \panel{\linewidth}
    {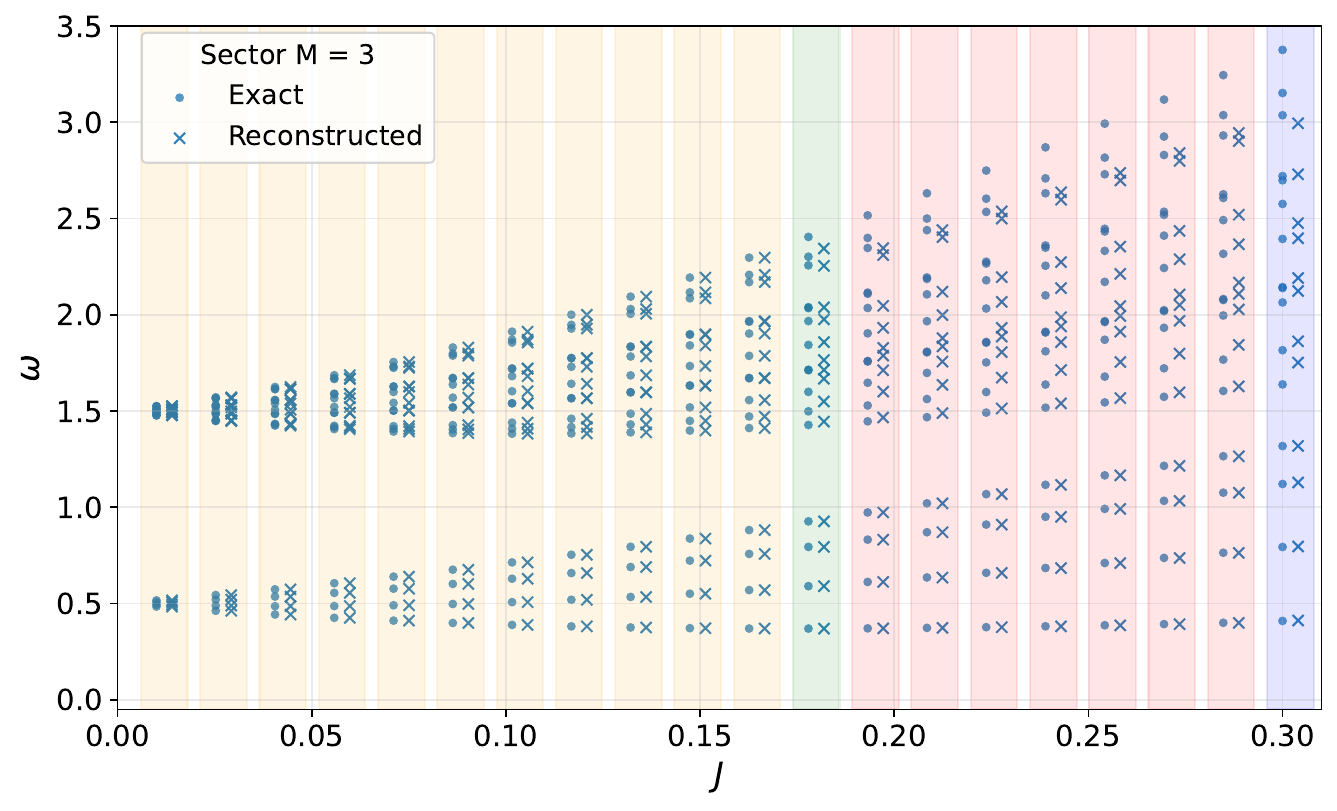}
    {(a)}{10pt}{0pt}

    \par\smallskip

    \panel{\linewidth}
    {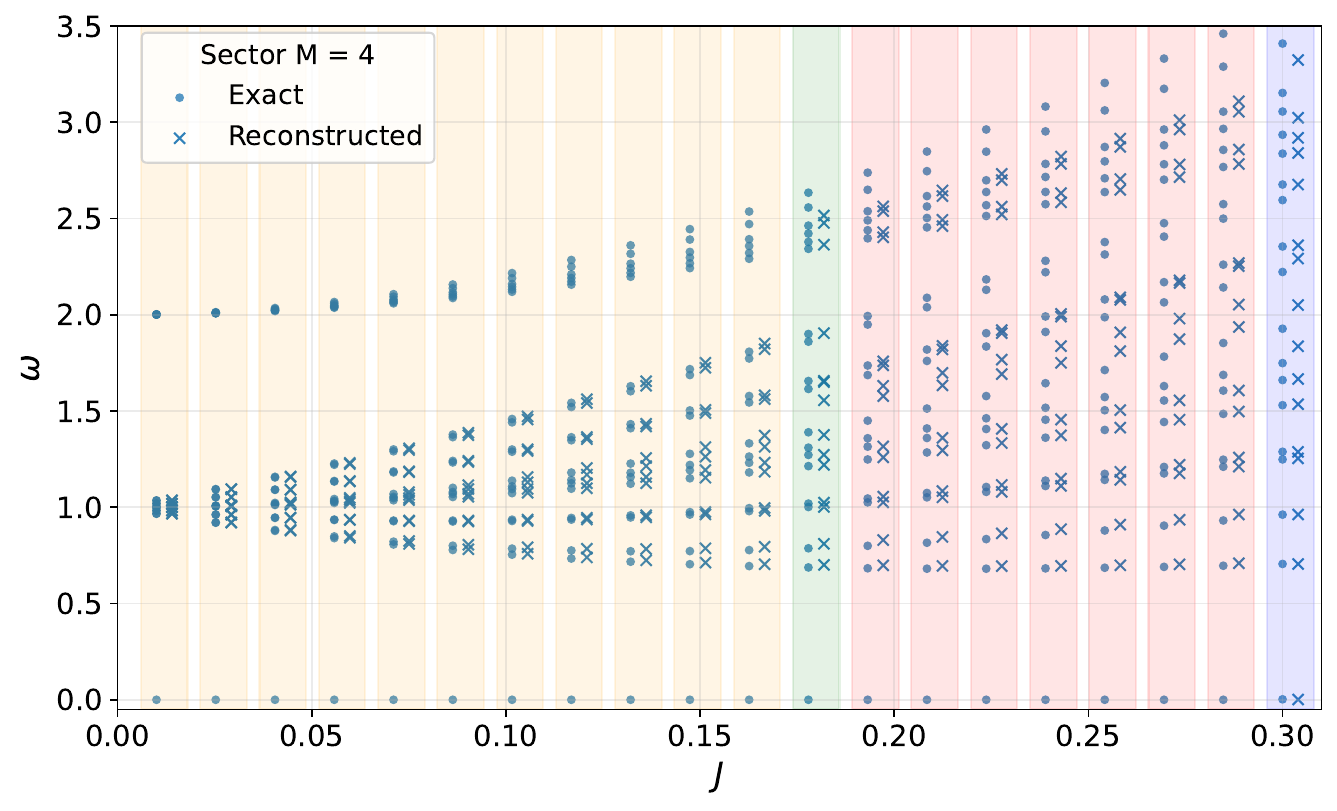}
    {(b)}{10pt}{0pt}

    \par\smallskip

    \panel{\linewidth}
    {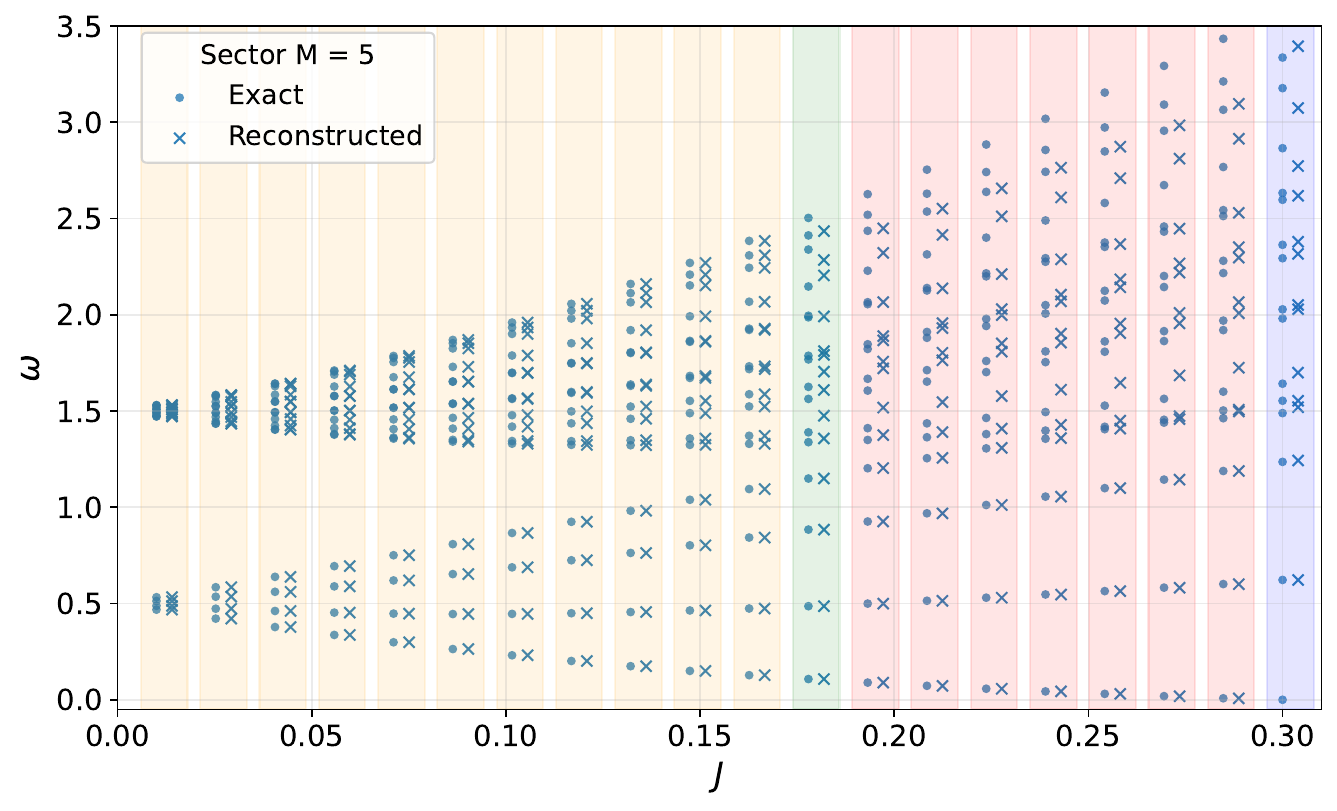}
    {(c)}{10pt}{0pt}

    \par\smallskip

    \makebox[\linewidth][l]{%
        \hspace*{-3pt}%
        \panel{1.01\linewidth}
        {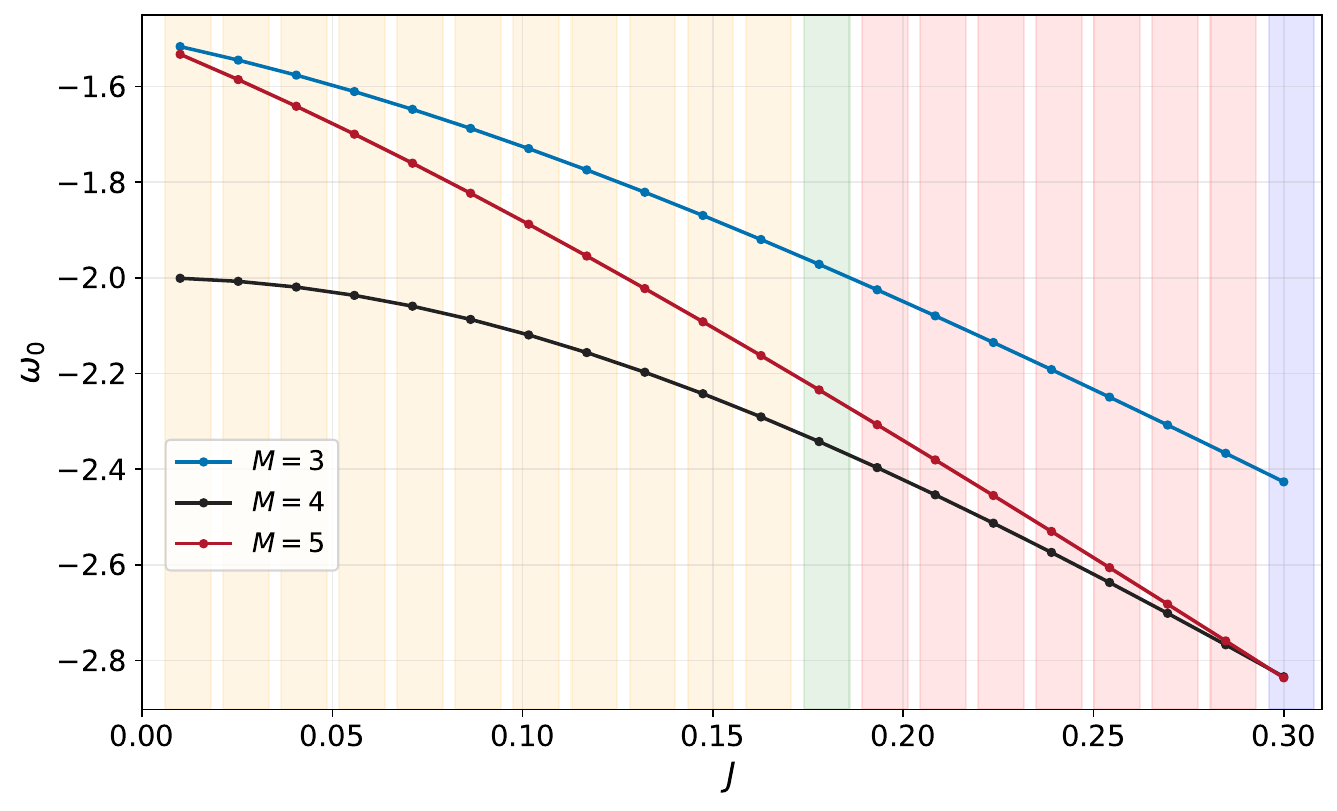}
        {(d)}{15pt}{0pt}
    }

    \caption{
    Exact excitation energies (dots) and tangent-space estimates
    (crosses) for Hamiltonian \eqref{eq:bose_hubbard} with
    $N=4$, $d=3$, $\mu/U=0.5$, $\eta=0$, and internal bond
    dimensions $(3,5,3)$.
    Panels (a)--(c) show excitations in sectors
    $M=3$, $M=4$, and $M=5$, respectively.
    Panel (d) shows the corresponding sector energies.
    Background shading labels the reference-state PRSR
    profiles in Table~\ref{tab:PRSR_four_site}.
    }
    \label{fig:N_4_d_3_mu_0.5}
\end{figure}

Table~\ref{tab:PRSR_four_site} lists the PRSR profiles across
cuts $\ell=2,3$, together with the complete tangent-rank
vectors and the deficiencies in the displayed sectors
$M=3,4,5$.
At the first cut $\ell = 1$, the PRSR profile does not change and is not shown. The profiles listed for cuts $\ell =2,3$  are ordered by $m=0,\dots,5$, with all other entries zero.








\begin{table*}[t]
\centering
\small
\setlength{\tabcolsep}{5pt}
\renewcommand{\arraystretch}{1.15}

\begin{tabular}{@{}lccccc@{}}
\toprule
Region
& $M_0$
& $\bigl(r_2(m)\bigr)_{m=0}^{5}$
& $\bigl(r_3(m)\bigr)_{m=0}^{5}$
& $\bigl(\rho_M\bigr)_{M=0}^{8}$
& $\bigl(\delta_M^{\mathrm{par}}\bigr)_{M=3}^{5}$
\\
\midrule

Orange
& $4$
& $(0,2,1,2,0,0)$
& $(0,0,1,1,1,0)$
& $(0,4,6,16,12,16,6,4,0)$
& $(0,6,0)$
\\[1pt]

Green
& $4$
& $(0,1,2,2,0,0)$
& $(0,0,1,1,1,0)$
& $(0,3,8,14,14,14,8,3,0)$
& $(2,4,2)$
\\[1pt]

Red
& $4$
& $(0,1,2,1,1,0)$
& $(0,0,1,1,1,0)$
& $(1,3,8,13,14,13,8,3,1)$
& $(3,4,3)$
\\[1pt]

Blue
& $5$
& $(0,1,2,1,1,0)$
& $(0,0,0,1,1,1)$
& $(0,2,6,12,16,13,10,4,1)$
& $(4,3,2)$
\\

\bottomrule
\end{tabular}

\caption{
Reference-state PRSR profiles and corresponding tangent-sector
dimensions for the four-site Bose--Hubbard example.
The region names refer to the background shading in
Figure~\ref{fig:N_4_d_3_mu_0.5}.
All references have internal ranks $(3,5,3)$, and
$\sum_M \rho_M = 64$ in every row.
The last column gives the deficiencies in sectors $M=3,4,5$,
using the reference particle number $M_0$ of that row. Note that for $M_0 = 4$ the tangent rank is always symmetric, while for $M_0=5$ it is not. This is an artifact of the local Hilbert space truncation to $d=3$.
}
\label{tab:PRSR_four_site}
\end{table*}

Figure~\ref{fig:N_4_d_3_mu_0.5} (a)-(c) compares the exact and
reconstructed spectra in sectors $M=3,4,5$.
These sectors make the effect of PRSR on spectrum reconstruction particularly visible.

In the orange region, the lower $M=4$ branch near
$\omega\simeq1$ at weak hopping is captured,
whereas the higher branch near $\omega\simeq2$ is not.
In the same interval, the displayed levels in sectors
$M=3,5$ agree with the exact results and the rank count identifies a precise structural difference.
Since $M_0=4$, the target dimensions are
\[
(d_3^{\mathrm{target}},
 d_4^{\mathrm{target}},
 d_5^{\mathrm{target}})
=
(16,18,16).
\]
For the orange profile, $(\rho_3,\rho_4,\rho_5)=(16,12,16)$,
and hence $(\delta_3^{\mathrm{par}},
 \delta_4^{\mathrm{par}},
 \delta_5^{\mathrm{par}})
=
(0,6,0)$.
Thus the tangent ansatz spans the entire sectors
$\mathcal H_3$ and $\mathcal H_5$, but misses six independent
directions from
$\mathcal H_4\cap\bar\psi^\perp$.

Complete sector coverage implies exact absolute sector
eigenvalues for the projected Hamiltonian.
The corresponding excitation estimates are still measured
relative to the variational reference energy $\omega_0$.
For the incomplete $M=4$ sector in the orange region, the dimension deficit
establishes a limitation of the trial space, but does not
by itself identify which individual modes will be inaccurate.

The comparison with the green profile makes the role of
rank allocation explicit. The total tangent dimension
remains $64$, but
\[
(\rho_3,\rho_4,\rho_5):
(16,12,16)\longrightarrow(14,14,14).
\]
The green profile has two fewer directions in each neighboring
sector and two more in the reference sector.
Its deficiencies are $(\delta_3^{\mathrm{par}},
 \delta_4^{\mathrm{par}},
 \delta_5^{\mathrm{par}})
=
(2,4,2)$.
Thus it loses complete coverage of $M=3,5$ while reducing
the deficit in $M=4$.
It is important to note that the ground state routine (DMRG) only seeks to minimize the energy given a (maximal) bond pattern $\boldsymbol{D}$ and does not impose any restrictions on the PRSR profile of the obtained ground state tensor. In fact, the energy is a perfectly continuous function of $J$ (see Figure \ref{fig:N_4_d_3_mu_0.5} (d)), while the PRSR exhibits a discrete pattern change between the orange and green region.

These completeness statements follow directly from the PRSR
profiles through \eqref{eq:PRSR_tangent_rank}; they do not
require construction or diagonalization of the tangent
Hamiltonian. The example therefore illustrates a practical
use of the diagnostic: comparing the excitation-sector
coverage of different reference-state rank allocations
at the same coarse bond dimensions.

\section{Conclusions}\label{sec:conclusion}

We have combined the established MPS tangent-space excitation
framework with a sector-resolved analysis of its expressivity
\cite{Haegeman_2014, Nakatani_2014}.
The algebro-geometric description by flattening-rank constraints
provides a parametrization-independent setting for the variational
space and its smooth full-rank stratum. Within this setting, we
formulate the projected-Hamiltonian excitation problem for finite,
non-uniform chains with open boundary conditions and examine the
rank restrictions on its trial space.

For a full-rank reference state with definite particle number,
we derive an explicit relation between its particle-resolved
Schmidt-rank profile (PRSR) and the dimensions of the
particle-number sectors of the physical tangent space.
Building on the established particle-number block structure
\cite{Bachmayr_2022}, this relation determines the parametric
deficiency---the number of missing independent directions in
each target sector---without constructing a tangent basis or
diagonalizing the effective Hamiltonian.
It explains why equal bond dimensions need not imply equal
sector coverage. 

Finite Bose--Hubbard chains provide a useful benchmark through
direct comparison with exact diagonalization.
The calculations accurately reproduce the low-lying branches
considered, both with and without particle-number conservation.
In the conserving case, the rank tomography identifies sector-dependent
restrictions and changes in expressivity at fixed bond dimensions,
complementing the spectral comparisons across ground-state
sector crossings.

A practical next step is to use the PRSR diagnostic to select
reference-state rank allocations for targeted excitation sectors.
Such calculations should assess improvements in absolute
excited-state energies separately from changes in the
reference energy used to form excitation gaps.
Further directions include extensions of the rank tomography to tree tensor networks and the explicit use of obtained eigenvectors for such systems.
These developments would test whether rank tomography can
guide the construction of excitation trial spaces, beyond
diagnosing their limitations.

\begin{acknowledgments}
We are grateful to Simon Telen, Matteo Rizzi, Christian Johansen, Alberto Biella, Alessandra Bernardi and Alberto Tabarelli de Fatis for continuous discussions. This work was supported by the European Union's Horizon
Europe programme under the Marie Sk{\l}odowska-Curie Actions,
Grant Agreement No.~101120296 (TENORS).
The authors also acknowledge financial support from the Provincia Autonoma di Trento; the Q@TN Initiative; the National Quantum Science and Technology Institute through the PNRR MUR Project under Grant PE0000023-NQSTI, co-funded by the European Union -- NextGeneration EU. 
\\
OpenAI GPT-5.4 was used during the preparation of this work for brainstorming research and presentation ideas. All suggestions were critically reviewed and verified by the authors.
\end{acknowledgments}

\appendix

\section{Block structure of local MPS tensors}
\label{app:block_form_matrices}
We illustrate how the PRSR determines the particle-number block
structure of the local MPS matrices. We consider variational
ground states of the Bose--Hubbard Hamiltonian
\eqref{eq:bose_hubbard} with $N=5$, $d=3$, $\mu=0.01$,
$\eta=0$, and maximum bond dimension $D=7$.
Along the sweep $J\in[0.01,0.3]$, the ground-state particle
number changes from $M_0=3$ to $M_0=4$ between the sampled
values $J=0.1931$ and $J=0.2084$.
The internal bond dimensions remain
\[
(D_1,D_2,D_3,D_4)=(3,7,7,3),
\qquad D_0=D_5=1.
\]

We examine the matrices at site $j=3$,
\[
M^3_n\in\mathbb C^{7\times7},
\qquad n=0,1,2,
\]
in the particle-number-resolved left-canonical representation
described in Section~\ref{sec:schmidt_rank_distribution}.
The rows label the states of sites $1$ and $2$, and are grouped by their
particle number $m$. The columns label the states of sites $1,2$ and $3$,
and are grouped by their particle number $m'$.
Consequently,
\[
M^3_n(m,m')\in
\mathbb C^{r_2(m)\times r_3(m')}.
\]
Since site $3$ contributes $n$ particles, the selection rule is
\[
M^3_n(m,m')=0
\qquad\text{unless}\qquad m'=m+n.
\]

In the schematic matrices below, $\star_{a\times b}$ denotes
a block of size $a\times b$ allowed by this selection rule.
Its numerical entries are not specified and need not all be
nonzero. Each $0$ denotes a zero block whose dimensions are
fixed by the corresponding row and column groups.
The symbol $\sim$ indicates the block pattern, rather than
numerical equality.

\paragraph{Before the sector crossing: $M_0=3$.}

At $J=0.1931$, the reference state lies in sector $M_0=3$.
The PRSR profiles are
\[
r_2=(1,2,3,1,0),
\qquad
r_3=(1,3,2,1,0),
\]
where the entries are ordered by particle number $0,\dots,4$;
all unlisted sectors have rank zero.
The rows are therefore grouped by $m=0,1,2,3$ with sizes
$(1,2,3,1)$, and the columns by $m'=0,1,2,3$ with sizes
$(1,3,2,1)$.

For occupation $n=0$, only groups with $m'=m$ can couple:
\[
M^3_0\sim
\begin{pmatrix}
\star_{1\times1} & 0 & 0 & 0\\
0 & \star_{2\times3} & 0 & 0\\
0 & 0 & \star_{3\times2} & 0\\
0 & 0 & 0 & \star_{1\times1}
\end{pmatrix}.
\]
For $n=1$, the allowed groups satisfy $m'=m+1$:
\[
M^3_1\sim
\begin{pmatrix}
0 & \star_{1\times3} & 0 & 0\\
0 & 0 & \star_{2\times2} & 0\\
0 & 0 & 0 & \star_{3\times1}\\
0 & 0 & 0 & 0
\end{pmatrix}.
\]
For $n=2$, they satisfy $m'=m+2$:
\[
M^3_2\sim
\begin{pmatrix}
0 & 0 & \star_{1\times2} & 0\\
0 & 0 & 0 & \star_{2\times1}\\
0 & 0 & 0 & 0\\
0 & 0 & 0 & 0
\end{pmatrix}.
\]

\paragraph{After the sector crossing: $M_0=4$.}

At $J=0.2084$, the reference state lies in sector $M_0=4$,
and the profiles are
\[
r_2=(1,2,2,2,0),
\qquad
r_3=(0,2,2,2,1).
\]
The row groups still carry $m=0,1,2,3$, but now have sizes
$(1,2,2,2)$. The column groups carry $m'=1,2,3,4$ and have
sizes $(2,2,2,1)$.
Both group sizes sum to seven, so the matrices remain
$7\times7$.

Applying the same selection rule gives
\[
M^3_0\sim
\begin{pmatrix}
0 & 0 & 0 & 0\\
\star_{2\times2} & 0 & 0 & 0\\
0 & \star_{2\times2} & 0 & 0\\
0 & 0 & \star_{2\times2} & 0
\end{pmatrix},
\]
\[
M^3_1\sim
\begin{pmatrix}
\star_{1\times2} & 0 & 0 & 0\\
0 & \star_{2\times2} & 0 & 0\\
0 & 0 & \star_{2\times2} & 0\\
0 & 0 & 0 & \star_{2\times1}
\end{pmatrix},
\]
and
\[
M^3_2\sim
\begin{pmatrix}
0 & \star_{1\times2} & 0 & 0\\
0 & 0 & \star_{2\times2} & 0\\
0 & 0 & 0 & \star_{2\times1}\\
0 & 0 & 0 & 0
\end{pmatrix}.
\]

The different positions of the allowed blocks reflect the
changed column labels, not a change in the conservation rule.
For example, the first row group has $m=0$ and the first column
group now has $m'=1$. Their block is therefore forbidden in
$M^3_0$ but allowed in $M^3_1$.

The comparison shows precisely what the coarse bond dimensions
do not record. Across cut $3$, the nonzero particle-number
groups change from $m'=0,1,2,3$ to $m'=1,2,3,4$:
the $m'=0$ group disappears and an $m'=4$ group appears, which makes fluctuations in some sectors (not) possible.
The multiplicities of the remaining groups also change.
These changes are invisible in the total matrix size
$7\times7$, but are recorded by the PRSR and determine the
sizes and locations of the allowed reference-state blocks.
The same row and column groups organize the tangent variations
used in the sector-dimension count.

\section{Particle-resolved coefficient matrices}
\label{app:coeff_matrix}
We illustrate the construction of the coefficient matrices
$\Psi_\ell^{(m)}$ for a four-site reference state with local
occupations $n_i=0,1,2$ and total particle number $M_0=4$.
We choose the cut after site $\ell = 2$ and write
\[
|n_1n_2\rangle_L|n_3n_4\rangle_R
:=
|n_1n_2n_3n_4\rangle,
\]
where the subscripts distinguish the left and right subsystems.

The ordered occupation bases of the left particle-number sectors are
\begin{align*}
\mathcal B_L(0)&=(|00\rangle_L),\\
\mathcal B_L(1)&=(|01\rangle_L,|10\rangle_L),\\
\mathcal B_L(2)&=(|02\rangle_L,|11\rangle_L,|20\rangle_L),\\
\mathcal B_L(3)&=(|12\rangle_L,|21\rangle_L),\\
\mathcal B_L(4)&=(|22\rangle_L).
\end{align*}
We use the same occupation ordering for the right bases
$\mathcal B_R(n)$.
Since the total particle number is four, the left sector $m$
pairs only with the right sector $4-m$.

Consider the reference state
\begin{align*}
\bar\psi
={}&
|00\rangle_L|22\rangle_R
+
|01\rangle_L|12\rangle_R
+
2|10\rangle_L|21\rangle_R\\
&+
|02\rangle_L|02\rangle_R
+
|11\rangle_L|11\rangle_R
+
|20\rangle_L|20\rangle_R\\
&+
|12\rangle_L
\bigl(|01\rangle_R+|10\rangle_R\bigr)\\
&+
|21\rangle_L
\bigl(|01\rangle_R-|10\rangle_R\bigr)\\
&+
|22\rangle_L|00\rangle_R.
\end{align*}
Every term has total particle number four.
We omit normalization to keep the coefficients simple.

Grouping the terms by the left particle number gives
\[
\bar\psi=\sum_{m=0}^{4}\bar\psi_2^{(m)},
\qquad
\bar\psi_2^{(m)}
\in
\mathcal H^{[1:2]}_m\otimes
\mathcal H^{[3:4]}_{4-m}.
\]
The matrix $\Psi_2^{(m)}$ records the coefficients of
$\bar\psi_2^{(m)}$, with rows ordered by $\mathcal B_L(m)$
and columns by $\mathcal B_R(4-m)$.
Reading these coefficients from the state gives
\begin{align*}
\Psi_2^{(0)}&=[1],
&
\Psi_2^{(1)}&=
\begin{pmatrix}
1&0\\
0&2
\end{pmatrix},\\
\Psi_2^{(2)}&=I_3,
&
\Psi_2^{(3)}&=
\begin{pmatrix}
1&1\\
1&-1
\end{pmatrix},\\
\Psi_2^{(4)}&=[1].
\end{align*}

For example, $\Psi_2^{(3)}$ has rows labelled by
$|12\rangle_L,|21\rangle_L$ and columns labelled by
$|01\rangle_R,|10\rangle_R$.
Its lower-right entry is $-1$, the coefficient of
$|21\rangle_L|10\rangle_R$.
Although this component contains four occupation-basis terms,
its coefficient matrix has rank two.
It therefore contributes two nonzero Schmidt coefficients,
not four.

All five matrices have full rank within their respective
particle-number blocks. Hence the PRSR profile is
\[
\bigl(r_2(m)\bigr)_{m=0}^{4}
=
\bigl(\rank\Psi_2^{(m)}\bigr)_{m=0}^{4}
=
(1,2,3,2,1).
\]
Ordering the left sectors as $m=0,\dots,4$ and the right
sectors in the complementary order gives
\[
\mathrm{Mat}_2(\bar\psi)
=
\bigoplus_{m=0}^{4}\Psi_2^{(m)}.
\]
The Schmidt rank across the cut, and therefore the minimal
MPS bond dimension at this cut, is
\[
D_2
=
\rank\bigl(\mathrm{Mat}_2(\bar\psi)\bigr)
=
\sum_{m=0}^{4}r_2(m)
=
1+2+3+2+1
=
9.
\]
Thus the block matrices determine both the total Schmidt
rank and its particle-number decomposition.
This example illustrates their explicit calculation;
Section~\ref{sec:examples_tomography} compares different
PRSR profiles at the same bond dimensions.

\section{Numerical implementation and tolerances}
\label{sec:numerical_details}

All calculations use double precision. We obtain the reference states
$\bar\psi$ by DMRG with maximum bond dimension $D$, 40 sweeps,
zero added noise, and zero truncation cutoff; truncation is controlled
by the bond cap. The optimized state at each $J$ initializes the next
point of the 20-point sweep. For $\eta=0$, the $N=4,10$ benchmarks
select the lower-energy reference from sectors $M=N,N+1$; the $N=3$
check uses $M_0=3$. For $\eta\ne0$, particle number is unconstrained.
The sweep count is fixed, with no additional energy- or
stationarity-based stopping criterion.

We compute $\ker R$ by SVD, treating singular values at or below
$10^{-12}$ as zero. MPS/MPO environment contractions give
$A=F^\dagger F$ and $B=F^\dagger H'F$, which are Hermitian-symmetrized.
In the Gram-matrix reduction of Appendix~\ref{sec:comp_costs}, we retain
only eigendirections satisfying
\[
\sigma_a^2>\max\bigl\{10^{-14},\,10^{-12}\|A\|_2\bigr\},
\]
where $\sigma_a^2$ are positive eigenvalues of $A$ and $\|A\|_2$ is its
spectral norm. The retained $V_s$ and $\Sigma$ yield
\[
B_{\mathrm{red}}
=\Sigma^{-1}V_s^\dagger B V_s\Sigma^{-1},
\]
which is diagonalized using a dense Hermitian eigensolver.
Its eigenvalues $\omega$ are measured relative to $\omega_0$.
These cutoffs determine numerical rank, not excitation-energy accuracy.

For $\eta=0$, reconstructed modes are labelled by the largest
weight $\|P_M\varphi\|^2$, with weights below $10^{-12}$ set to zero.
The PRSR profiles are read from particle-number multiplicities of the
canonical MPS bond indices, identified with $r_\ell(m)$ under the
full-rank assumption; no separate Schmidt-value cutoff is applied.
Spectrum plots omit $\omega<-10^{-10}$; the accuracy comparisons
use the unfiltered lowest levels in the selected sectors.

Exact diagonalization uses the same local truncation and boundaries.
For $\eta=0$, all particle-number sectors are included, retaining up
to 56 lowest levels per sector; the symmetry-breaking benchmark
retains the 100 lowest levels of the full Hamiltonian.
Matrices of dimension at most 256 are diagonalized densely;
larger matrices use a sparse Hermitian eigensolver with tolerance
$10^{-10}$. Exact excitation energies are measured relative to the
global exact ground-state energy $E_0$ at each $J$.

OpenAI GPT-5.4 was used to assist with the generation and debugging of parts of the numerical code. All AI-generated code was reviewed, tested, and modified as necessary by the authors, who take full responsibility for the implementation and reported results.

\section{Computational costs}
\label{sec:comp_costs}
We estimate the cost of the spectrum reconstruction in
Section~\ref{sec:linear_tangent_mps}, excluding the computation
of the variational ground state. Let
\[
q:=\prod_{j=1}^N d_j,
\qquad
P:=\sum_{j=1}^N d_jD_{j-1}D_j
\]
denote the Hilbert-space dimension and the number of
unconstrained complex core parameters, respectively.
Core variations are represented by $2P$ real coordinates.
The horizontal dimension is
$r=\dim_{\mathbb R}\ker R\leq2P$, and the complete physical
tangent space has complex dimension $s=r/2$ under the
full-rank assumptions.

In Section \ref{sec:linear_tangent_mps}, we orthonormalize the tangent vectors by a singular value decomposition of $F$. Explicitly forming $F$, however, requires storing $q$-dimensional Hilbert-space vectors. The same orthonormalization can be performed implicitly from pairwise tangent-vector overlaps. Define
\[
A_{ab}:=\langle\Theta_a,\Theta_b\rangle,
\qquad
B_{ab}:=\langle\Theta_a,H'\Theta_b\rangle.
\]
These $r\times r$ matrices can be evaluated directly using MPS/MPO contractions. Since $A=F^\dagger F$, its nonzero eigendecomposition $A=V_s\Sigma^2V_s^\dagger$
determines the singular values and right singular vectors of $F$. The corresponding orthonormal tangent basis would be $U=FV_s\Sigma^{-1}$, but need not be formed explicitly. Instead,
\[
U^\dagger H'U
=
\Sigma^{-1}V_s^\dagger B V_s\Sigma^{-1}.
\]
Hence both the tangent-space orthonormalization and the reduced Hamiltonian can be obtained entirely from $r\times r$ matrices.

For homogeneous local dimensions $d_j=d$ and $D_j\leq D$,
we have $P,r=\mathcal O(NdD^2)$ and let $w$ be the Hamiltonian's MPO bond dimension.
Assuming $\mathcal O(w)$ nonzero MPO virtual blocks per site
and environment sweeps that accumulate the tangent-insertion
sums, we use the conservative pair-contraction estimate $\mathcal O\!\left(Nw d^2D^4\right)$.
Evaluating $\mathcal O(r^2)$ pairs to assemble $A$ and $B$
therefore costs $\mathcal O\!\left(r^2Nw d^2D^4\right)$.
This estimate depends on the contraction strategy and is
not an optimal bound for all implementations.

For a dense constraint matrix $R$ with $2P$ columns and
$\mathcal O(P)$ rows, computing the horizontal basis costs
$\mathcal O(P^3)$.
Diagonalizing $A$, forming $B_{\mathrm{red}}$, and solving
the reduced eigenproblem together cost $\mathcal O(r^3)$.
The total operation count is therefore
\[
\mathcal O\!\left(N^3wd^4D^8+N^3d^3D^6\right).
\]
with $\mathcal O(N^2d^2D^4)$ dense-array memory.
The implementation thus avoids explicit storage in the
full Hilbert space and has polynomial cost at prescribed
bond dimensions. Maintaining a fixed spectral accuracy
may nevertheless require increasing these dimensions
with system size.

\section{Exact-limit consistency check}
\label{sec:exact_spectrum}
An exact-limit check is obtained when the MPS variety fills
the ambient projective Hilbert space, $V_{\boldsymbol D,\boldsymbol d}=\mathbb P(\mathcal H)$.
At a full-rank reference state, the physical tangent space is
then $T_{\bar\psi}=\bar\psi^\perp$,
so no physical tangent directions are excluded.
If $\bar\psi$ is an exact normalized ground state,
$H\bar\psi=E_0\bar\psi$, then $\omega_0=E_0$ and
$\bar\psi^\perp$ is invariant under $H$.
Consequently, the reduced Hamiltonian
\[
B_{\mathrm{red}}=U^\dagger(H-E_0\id)U
\]
has the complete excitation spectrum, excluding only the
chosen reference-state direction.

We illustrate this limit for the Bose--Hubbard model with
$N=3$, $d=3$, $\mu=0.5$, $U=1$, and internal bond dimensions $(D_1,D_2)=(3,3)$ and $D_0=D_3=1$.
The two flattenings have sizes $3\times9$ and $9\times3$,
so their rank bounds are automatically satisfied by every
state. Hence, $V_{\boldsymbol D,\boldsymbol d}
=
\mathbb P(\mathbb C^{27})$ and $\dim_{\mathbb C}T_{\bar\psi}=26$.
The tangent ansatz therefore contains all 26 directions
orthogonal to a full-rank exact reference ground state.

Figure~\ref{fig:variety_filling} compares the exact and
reconstructed excitation spectra in this setting.
Reconstructed and exact spectrum coincide up to machine precision. 

\begin{figure}[ht]
    \centering
    \includegraphics[width=\linewidth]
    {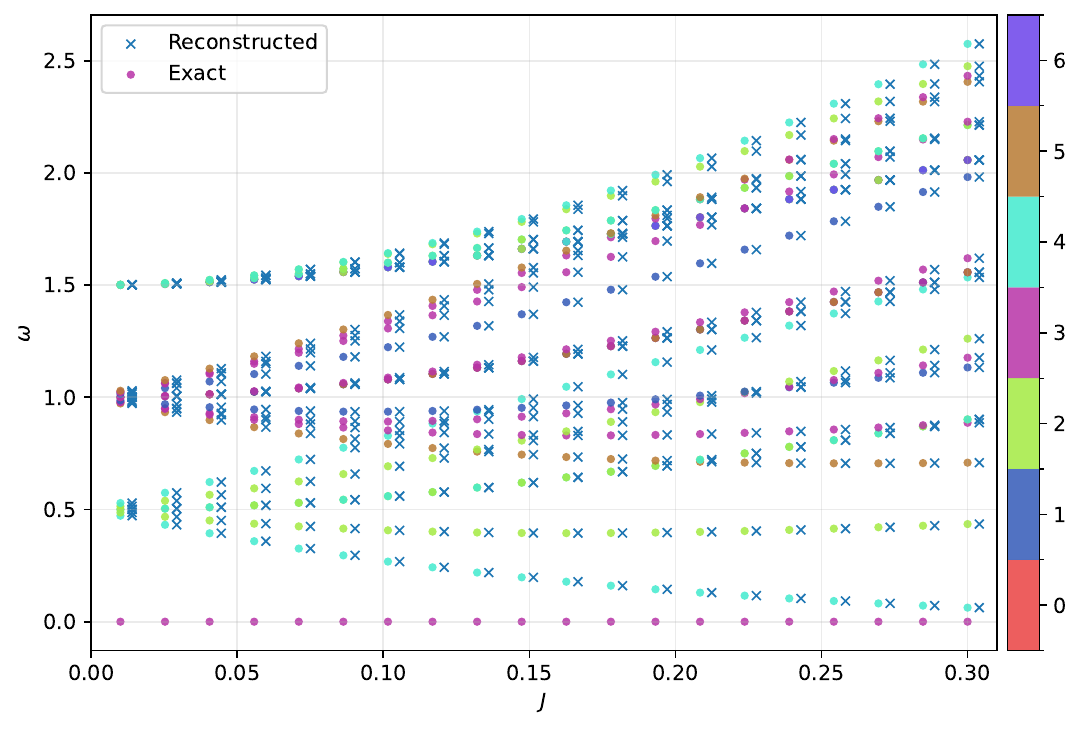}
    \caption{
    Exact excitation energies (coloured dots) and tangent-space
    estimates (blue crosses) for a three-site Bose--Hubbard
    chain with $d=3$, $\mu/U=0.5$, $\eta=0$, and internal bond dimensions
    $(3,3)$. Colours identify the particle-number sectors
    of the excitation states.
    These bond dimensions allow the MPS variety to fill
    the ambient projective Hilbert space, providing an
    exact-limit consistency check.
    }
    \label{fig:variety_filling}
\end{figure}

\section{Proofs}
\label{app:proofs}

\subsection{Proof of Proposition~\ref{prop:smoothness_full_rank}}
\label{app:proofs_smoothness_full_rank}
\begin{proof}
By Proposition \ref{prop:rank_description_variety}, $V_{\boldsymbol D,\boldsymbol d}$ is cut out by homogeneous minors
of the flattenings and is a projective algebraic variety. For fixed $k$, inside $V_{\boldsymbol D,\boldsymbol d}$ the condition $\rank(T^{(k)})= D_k$ is equivalent to saying that at least one $D_k\times D_k$ minor of $T^{(k)}$ is non-zero. Equivalently, it is the complement of the closed condition $\rank(T^{(k)})\le D_k-1.$
Hence,
\begin{equation*}
    \{[T]\in V_{\boldsymbol D,\boldsymbol d}
\mid \rank(T^{(k)})=D_k\}
\end{equation*}
is Zariski open in $V_{\boldsymbol D,\boldsymbol d}$. Taking the finite
intersection over $k=1,\dots,N-1$, we get that
$V^{=}_{\boldsymbol D,\boldsymbol d}$ is Zariski open in
$V_{\boldsymbol D,\boldsymbol d}$.
\\
It remains to prove smoothness. For this, we first make a general observation. For a matrix
$M$ of rank $r$, choose an invertible $r\times r$ minor $M_{I,J}$, where $I,\,J$ are collections of indices for rows and columns, respectively, defining the minor. On the open set where this minor remains non-zero, $M$ has the unique
gauge-fixed factorization $M=AB$: define $A:=M_{*,J}M_{I,J}^{-1}$ and $B:=M_{I,*}$, where $M_{*,J}$ are the columns $J$ and $M_{I,*}$ are the rows $I$. The gauge freedom of this factorization is removed by requiring $A_{I,*}=I_r$. Thus, the rank-\(r\) matrix locus is isomorphic to 
\begin{equation*}
    \{(A,B): A_{I,*} = I_r, ~\det(B_{*, J})\neq 0\},
\end{equation*}
which is a Zariski open subset of an affine space. We can conclude that the variety of matrices $M$ of rank $r$ form a smooth variety. 
\\
Now, define the affine full-rank locus
\begin{equation*}
    \widetilde V^{=}_{\boldsymbol D,\boldsymbol d}
=
\left\{
T\in\mathcal H
\;\middle|\;
\rank(T^{(k)})=D_k,\; k=1,\dots,N-1
\right\}
\end{equation*}
and consider a tensor $T\in \widetilde V^{=}_{\boldsymbol D,\boldsymbol d}$.
\\
We now use the smoothness of a full-rank determinantal variety that was given above and apply this successively to the flattenings of $T$. First factor $T^{(1)}$ with rank $D_1$, such that $T^{(1)}=A_1T_2^{(1)}$, with $A_1\in\C^{d_1\times D_1}$ and $T_2 \in \C^{D_1} \otimes \C^{d_2} \otimes \cdots \otimes \C^{d_N}$. Since $A_1$ has full column rank, the later flattening ranks are unchanged:
\begin{equation*}
    \rank(T^{(k)})=\rank(T_2^{(k-1)}),
\qquad k=2,\dots,N-1.
\end{equation*}
Now, take $T_2$ and flatten it to shape $T_2 \in \C^{D_1d_2 \times d_3\cdots d_N}$. Therefore, $T_2^{(1)}$ is a matrix of rank $D_2$ and we write again $T_2^{(1)} = A_2T_3^{(1)}$ with $A_2\in\C^{d_2D_1\times D_2}$ and $T_3 \in \C^{D_2} \otimes \C^{d_3} \otimes \cdots \otimes \C^{d_N}$. Repeating the same construction, we find that the free coordinates of $T \in \widetilde V^{=}_{\boldsymbol D,\boldsymbol d}$ are the entries of the tensors $\{A_1 ,\dots, A_{N-1}, T_N\}$, where $A_i\in\C^{(D_{i-1}d_i)\times D_i}$ satisfies a gauge condition $(A_i)_{I_i,*}=I_{D_i}$ and the final tensor $T_N\in\C^{D_{N-1}\times d_N}$ has full row rank. The remaining requirements are non-vanishing minor conditions, hence open conditions. Therefore, $\widetilde V^{=}_{\boldsymbol D,\boldsymbol d}$ is locally isomorphic to an
open subset of affine space, and hence smooth.
\\
Finally, the projectivization $V^{=}_{\boldsymbol D,\boldsymbol d}
=
\P\bigl(\widetilde V^{=}_{\boldsymbol D,\boldsymbol d}\bigr)$ does not affect smoothness. Therefore
$V^{=}_{\boldsymbol D,\boldsymbol d}$ is a smooth Zariski open subset of $V_{\boldsymbol D,\boldsymbol d}$.
\end{proof}
\vspace{1em}

\subsection{Proof of Proposition~\ref{prop:horizontal_conditions}}
\label{app:proofs_horizontal_conditions}
\begin{proof}
Let $\delta p=(\delta M,\delta C)$ be a parameter tangent
vector, with $\|C\|_F=1$ and
$\operatorname{tr}(C^\ast\delta C)=0$.
Using the real Frobenius metric and the chosen orthogonal
representatives for the projective factor, its pairing with
a vertical vector $\dot p=(\dot M,\dot C)$ is
\begin{align*}
g_p(\delta p,\dot p)
={}&
\sum_{i=1}^{N-1}\sum_{j_i=1}^{d_i}
\operatorname{Re}\operatorname{tr}
\bigl((\delta M^i_{j_i})^\ast\dot M^i_{j_i}\bigr)\\
&+
\operatorname{Re}\operatorname{tr}
\bigl(\delta C^\ast\dot C\bigr).
\end{align*}
Any correction to $\dot C$ proportional to $C$ has zero
pairing with $\delta C$.
Thus, in this expression, the projective gauge variation
may be represented by $\dot C=-X_{N-1}C$.

Substituting the vertical variations from
\eqref{eq:vertical_directions_explicit} and using cyclicity
of the trace gives
\[
g_p(\delta p,\dot p)
=
\sum_{i=1}^{N-1}
\operatorname{Re}\operatorname{tr}(X_iY_i),
\]
where the matrices $Y_i$ are those in the proposition.
Since the skew-Hermitian matrices $X_i$ can be chosen
independently, horizontality is equivalent to
\[
\operatorname{Re}\operatorname{tr}(X_iY_i)=0
\qquad
\text{for every }X_i\in\mathfrak u(D_i)
\]
at each bond.

For any complex matrix $Y$, write
\[
Y=Y_{\mathrm h}+Y_{\mathrm s},
\qquad
Y_{\mathrm h}^\ast=Y_{\mathrm h},
\qquad
Y_{\mathrm s}^\ast=-Y_{\mathrm s}.
\]
The real trace pairing of a skew-Hermitian matrix with
$Y_{\mathrm h}$ vanishes.
Choosing $X=Y_{\mathrm s}^\ast$, which is skew-Hermitian,
therefore gives
\[
\operatorname{Re}\operatorname{tr}(XY)
=
\operatorname{tr}(Y_{\mathrm s}^\ast Y_{\mathrm s})
=
\|Y_{\mathrm s}\|_F^2.
\]
Hence the pairing vanishes for every skew-Hermitian $X$
if and only if $Y_{\mathrm s}=0$, or equivalently,
$Y$ is Hermitian.
Applying this to each $Y_i$ proves the claim.
\end{proof}

\bibliography{apssamp}

\end{document}